\documentclass[a4paper,onecolumn, ,accepted=2024-09-02]{quantumarticle}

\pdfoutput=1
\usepackage[utf8]{inputenc}
\usepackage[english]{babel}
\usepackage[T1]{fontenc}
\usepackage[pdftex]{hyperref}
\hypersetup{
    pdftitle = {The Hadamard gate cannot be replaced
    by a resource state in universal quantum computation}
}

\usepackage[numbers,sort&compress]{natbib}

\usepackage[braket]{qcircuit}
\newcommand{\qvdots}{
  \raisebox{0.3em}{\ensuremath{\vdots}}%
}

\usepackage{amsmath,amssymb,amsthm,bm,mathtools,amsfonts,mathrsfs,bbm,dsfont, physics}
\usepackage[shortlabels]{enumitem}

\usepackage{relsize}

\allowdisplaybreaks

\usepackage{tikz}
\usepackage{float}
\usepackage[center]{caption}
\usepackage{subcaption}
\usetikzlibrary{positioning}
\usetikzlibrary{shapes}
\usepackage{changepage}
\usepackage{graphicx}

\usepackage{arydshln}

\usepackage{pifont}

\usepackage[framemethod=TikZ]{mdframed} 
\usepackage{thm-restate} 

\usepackage[capitalize]{cleveref}

\DeclareRobustCommand{\[}{\begin{equation}}
\DeclareRobustCommand{\]}{\end{equation}}

\newcommand{\id}{\ensuremath{\mathds{1}}}

\declaretheoremstyle[
    headfont=\bfseries, 
    bodyfont=\normalfont,
    headpunct={.},
    spacebelow=\parsep,
    spaceabove=\parsep,
    mdframed={
      roundcorner=10pt,
      linecolor=quantumviolet, 
      linewidth=1pt,
        innertopmargin=6pt,
        innerbottommargin=6pt, 
        skipabove=3ex, 
        skipbelow=3ex,
       } 
]{framedstyle}

\declaretheoremstyle[
    headfont=\bfseries, 
    bodyfont=\normalfont,
        headpunct={},
    spacebelow=\parsep,
    spaceabove=\parsep,
    mdframed={
      roundcorner=10pt,
      linecolor=quantumviolet, 
      linewidth=1pt,
        innertopmargin=6pt,
        innerbottommargin=6pt, 
        skipabove=3ex, 
        skipbelow=3ex,
       } 
]{unnamedstyle}

\declaretheorem[style=framedstyle,name=Theorem]{theorem}

\declaretheorem[style=framedstyle,name=Lemma, numberlike=theorem]{lemma}
\declaretheorem[style=framedstyle,name=Definition, numberlike=theorem]{definition}
\declaretheorem[style=framedstyle,name=Corollary, numberlike=theorem]{corollary}
\declaretheorem[style=framedstyle,name=Observation, numberlike=theorem]{observation}
\declaretheorem[style=framedstyle,name=Remark, numberlike=theorem]{remark}
\declaretheorem[style=framedstyle,name=Example, numberlike=theorem]{example}

\declaretheorem[name=Problem]{problem}
\declaretheorem[style=unnamedstyle,name=, numbered=no]{nonamethm}

\makeatletter
\renewcommand\tableofcontents{%
    \@starttoc{toc}%
}
\makeatother

\newcommand{\ancillapure}{\gamma}
\newcommand{\ancillaket}{\ket{\ancillapure}}
\newcommand{\ancillaproj}{\ketbra{\ancillapure}}
\newcommand{\ancillamixed}{\tau}

\title{The Hadamard gate cannot be replaced
by a resource state in universal quantum computation }

\author{Benjamin D.M. Jones}
\affiliation{H. H. Wills Physics Laboratory, University of Bristol, Bristol, BS8 1TL, UK.}\affiliation{School of Mathematics, University of Bristol, Fry Building, Woodland Road, Bristol, BS8 1UG, UK.}\affiliation{Quantum Engineering Centre for Doctoral Training,  University of Bristol, Bristol, BS8 1FD
UK.}
\author{Noah Linden}
\affiliation{School of Mathematics, University of Bristol, Fry Building, Woodland Road, Bristol, BS8 1UG, UK.}
\author{Paul Skrzypczyk}
\affiliation{H. H. Wills Physics Laboratory, University of Bristol, Bristol, BS8 1TL, UK.}
\affiliation{CIFAR Azrieli Global Scholars Program, CIFAR, Toronto Canada.}

\begin{document}

\maketitle

\begin{abstract}
    We consider models of quantum computation that involve operations performed on some fixed resourceful quantum state. Examples that fit this paradigm include magic state injection and measurement-based approaches. We introduce a framework that incorporates both of these cases and focus on the role of coherence (or superposition) in this context, as exemplified through the Hadamard gate. We prove that given access to incoherent unitaries (those that are unable to generate superposition from computational basis states, e.g. CNOT, diagonal gates), classical control, computational basis measurements, and any resourceful ancillary state (of arbitrary dimension), it is not possible to implement any coherent unitary (e.g. Hadamard) exactly with non-zero probability. We also consider the approximate case by providing lower bounds for the induced trace distance between the above operations and $n$ Hadamard gates. To demonstrate the stability of this result, this is then extended to a similar no-go result for the case of using $k$ Hadamard gates to exactly implement $n>k$ Hadamard gates.
\end{abstract}

  \tableofcontents

\section{Introduction}

 \begin{figure}[b!]
        \centering 
         \begin{subfigure}[b]{0.24\textwidth}   
            \centering \label{fig:nocut}
            \includegraphics[scale=0.15, trim = 1.5cm 1.5cm 1.5cm 1.5cm ]{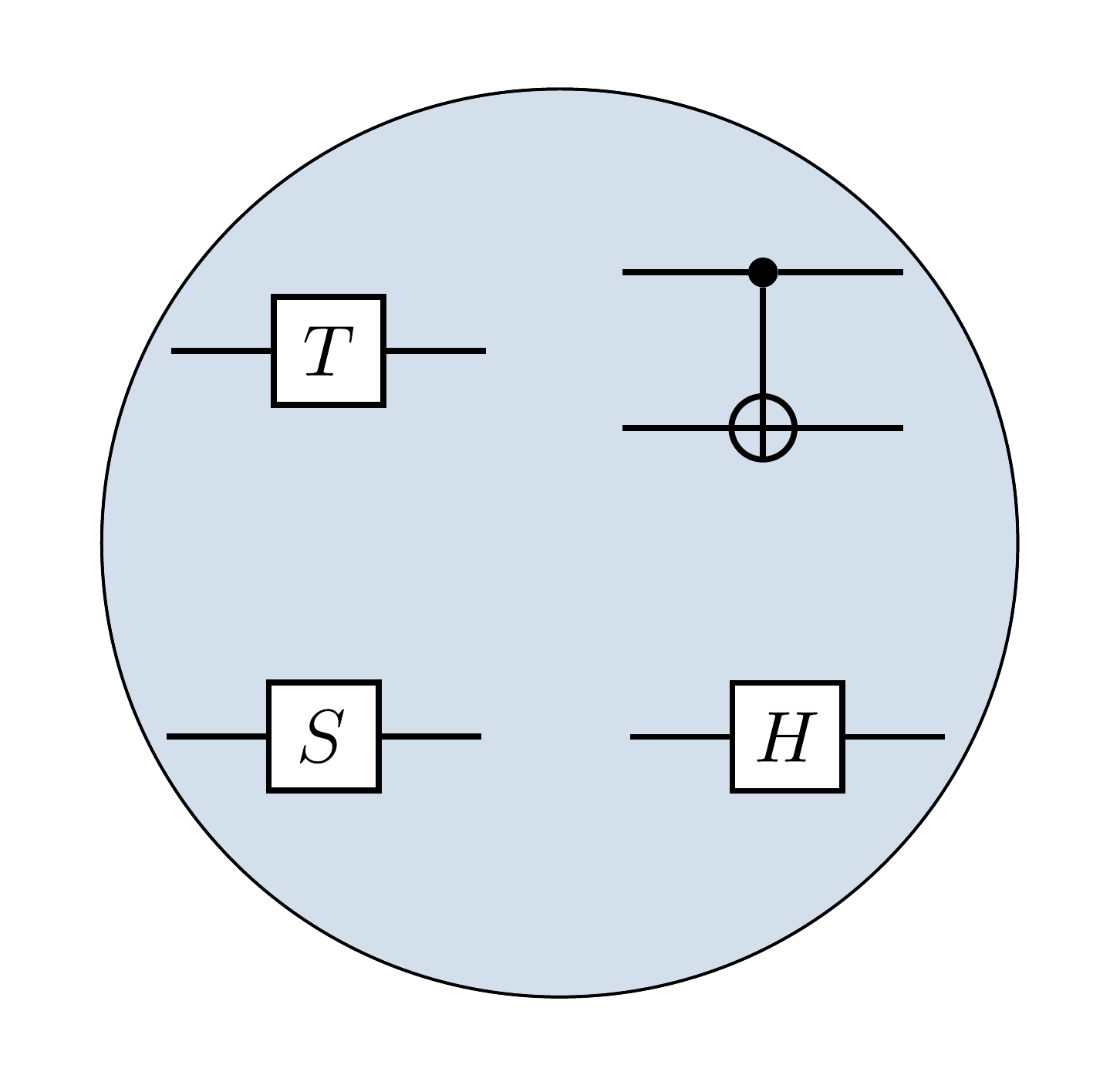}
            \caption{Universal gate set.}
        \end{subfigure}
        \hfill
         \begin{subfigure}[b]{0.24\textwidth}   
            \centering 
            \includegraphics[scale=0.15, trim = 1.5cm 1.5cm 1.5cm 1.5cm]{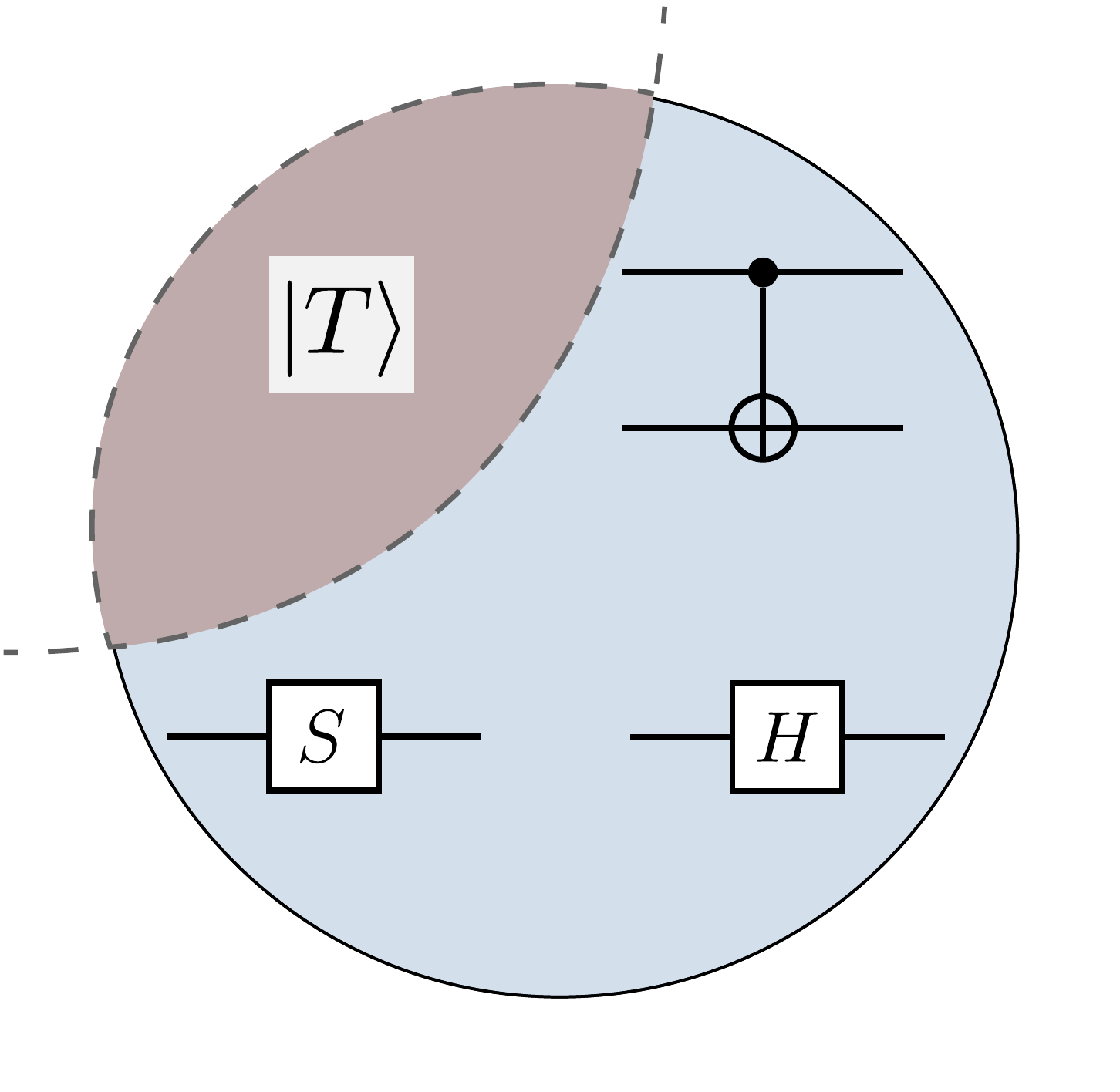}
            \caption{Magic state injection.}
        \end{subfigure}
        \hfill
         \begin{subfigure}[b]{0.24\textwidth}   
            \centering 
            \includegraphics[scale=0.15, trim = 1.5cm 1.5cm 1.5cm 1.5cm]{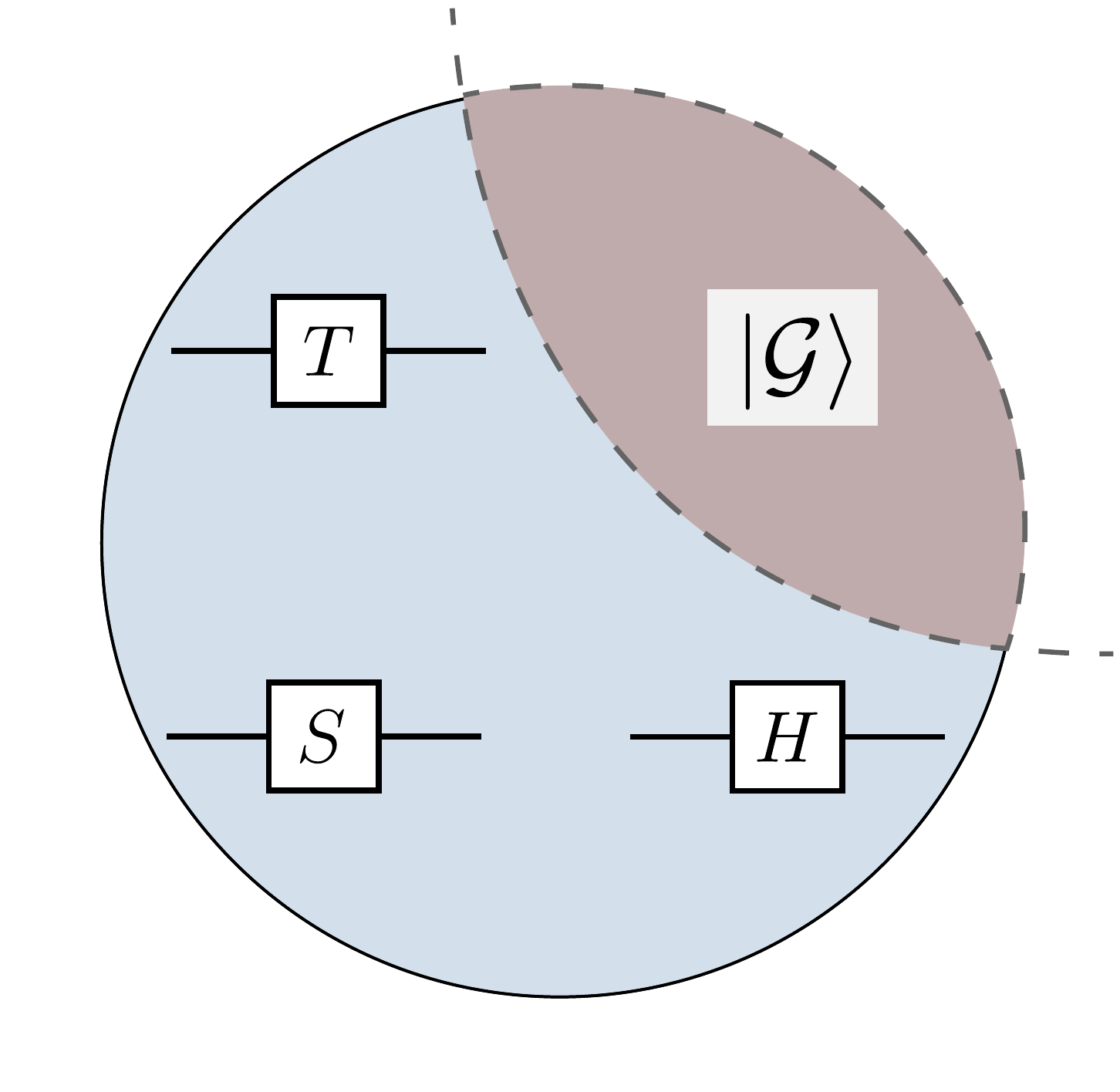}
            \caption{MBQC.}
        \end{subfigure}
        \hfill
        \begin{subfigure}[b]{0.24\textwidth}   
            \centering 
            \includegraphics[scale=0.15, trim = 1.5cm 1.5cm 1.5cm 1.5cm]{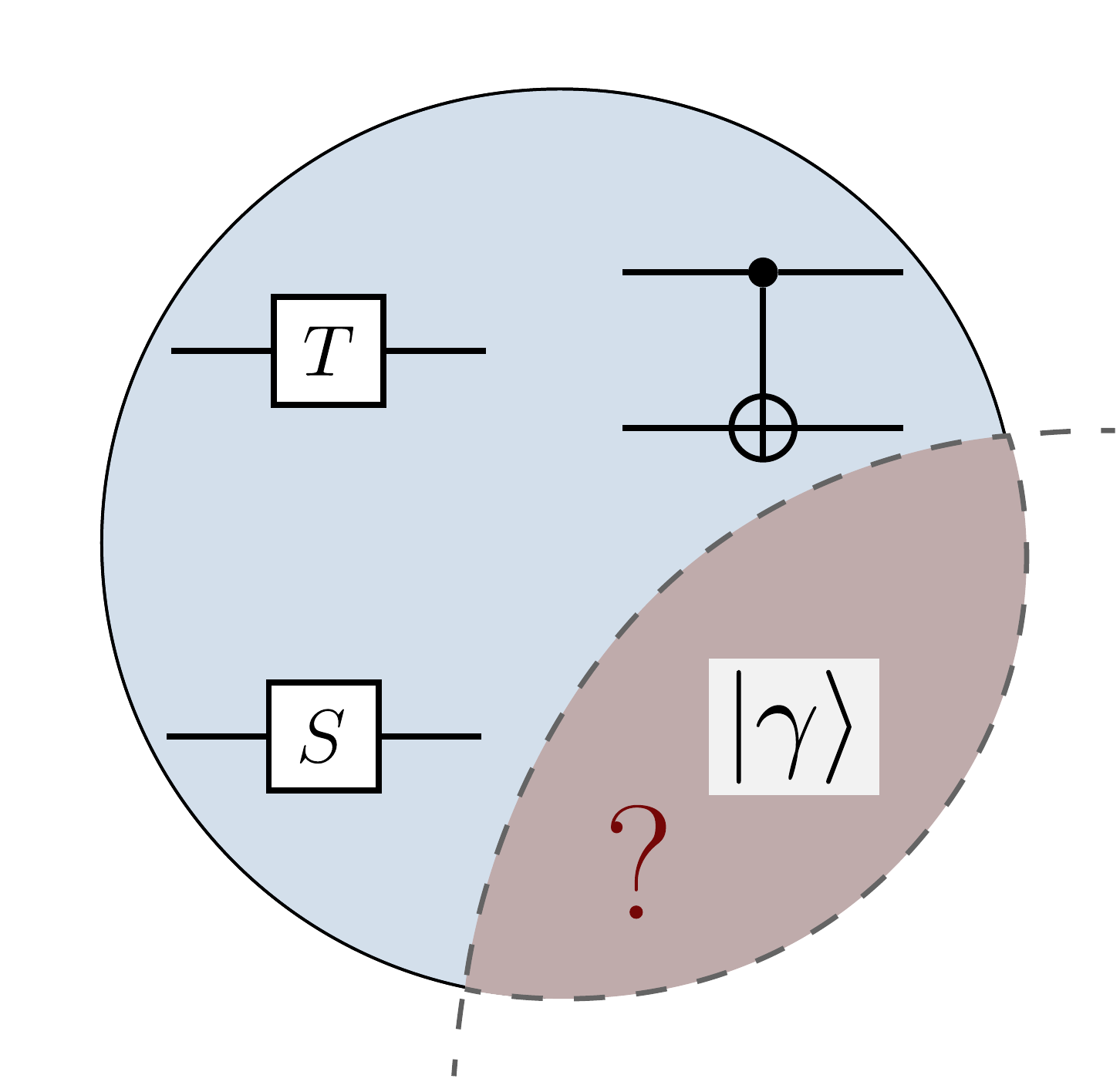}
            \caption{This work.}
        \end{subfigure}
        \caption{(a) A universal set of quantum gates: Hadamard ($H$), Phase ($S$), $T$, and controlled-NOT. (b) Each $T$ gate can be implemented using Clifford operations and a $\ket{T}$ state. (c) If CNOT is removed from this gate set one can still perform universal quantum computation using an appropriately entangled resource state $\ket{\mathcal{G}}$, such as a cluster state. (d) We ask whether one can similarly replace the Hadamard gate with some resourceful state $\ket{\gamma}$ and still achieve universality -- we show that this is not possible. In all cases we allow computational basis measurements and classical control. Also note that the case of removing $S$ is trivial as $T^2 = S$.}
        \label{fig:clifford+T cut}
    \end{figure}

The more peculiar aspects of quantum mechanics, such as entanglement \cite{horodecki2009quantum} and incompatibility of measurements \cite{guhne2021incompatible}, continue to fascinate researchers and motivate a deeper understanding of this cornerstone of physics. It is also remarkable that quantum theory appears to provide a computational speed-up for certain problems over what is possible with classical physics \cite{nielsen2002quantum}. The quest to fully understand and quantify which aspects of quantum theory are needed for useful quantum algorithms is a pressing and exciting current area of research.

There are various ways of performing universal quantum computation: examples include the circuit model \cite{nielsen2002quantum}, measurement based approaches \cite{jozsa2006introduction}, magic state injection \cite{campbell2017roads}, quantum annealing \cite{albash2018adiabatic}, and continuous variable models \cite{braunstein2005quantum}. An interesting perspective is to consider approaches involving ``free'' operations (i.e. easy to perform in some sense) acting on a resourceful state that is prepared independently of the computation. By focusing on this supplementary state, one could hope to gain insight into which components of quantum mechanics are responsible for the computational classical-quantum boundary. 

The most widely studied universal gate set is the Clifford + $T$ gate set; recall that the Clifford group is generated by the single qubit Hadamard ($H$) and phase ($S$) gates and the two-qubit controlled-NOT (CNOT) gate. The gate set of CNOT, $T$ and Hadamard is also universal, and can be thought of as respectively supplying the resources of entanglement, magic (or non-stabiliserness) and coherence (or superposition). In magic state injection (MSI), one implements a $T$ gate by performing adaptive Clifford operations on the input state and an ancillary state $\ket{T}:= T\ket{+}$. Here the operations performed are free with respect to the resource of magic, and all of the magic required is contained in the pool of $\ket{T}$ states. This approach is motivated by error correction and fault tolerance schemes \cite{campbell2017roads}. 

In contrast, measurement based quantum computation (MBQC) proceeds by adaptively performing single qubit measurements on an entangled resource state, such as a cluster state \cite{van2007fundamentals}. In this scenario, the resource of entanglement is present only in the state, and again the operations are free with respect to this resource. The ability to perform computational basis measurements (i.e. measure in the $Z$ basis) and apply $H$, $S$ and $T$ gates also implies ability to measure in the $X$, $Y$ and $TXT^\dagger$ bases, which is sufficient for universality \cite{takeuchi2019quantum}.

From these examples, a natural question arises of where we can put the `cut' between operations and states whilst retaining the ability to perform universal quantum computation -- see \cref{fig:clifford+T cut}. For example when considering the Clifford + $T$ gate set, can one replace the Hadamard with access to some resourceful state, and still maintain universality? We provide no-go results in this direction. 

In more generality one can consider whether this cut is possible for an arbitrary quantum resource theory \cite{chitambar2019quantum}. As Hadamard is the only gate within Clifford + $T$ capable of generating superpositions from computational basis states, the relevant resource theory here is that of coherence \cite{streltsov2017colloquium}. Our findings show that some coherence is required in the \emph{operations} to achieve universality, providing a stark contrast with the resource theories of magic and entanglement.

\subsection{Summary of Results}

We provide no-go results on the possibility of performing universal quantum computation using operations unable to generate superpositions, even given access to an arbitrary state. A unitary that maps at least one computational basis state to a superposition of two or more basis states is termed \textit{coherent}, otherwise it is \textit{incoherent}. Our findings can be informally summarised as:
\begin{nonamethm}
 $$\let\scriptstyle\textstyle
 \substack{\text{Incoherent} \\ \text{unitaries}} \quad + \quad\substack{\text{classical} \\ \text{control}} \quad + \quad \substack{\text{computational basis} \\ \text{measurements}}  \quad + \quad \substack{\text{arbitrary} \\ \text{ancillas}}
 $$
 ~
  $$
\text{cannot implement coherent unitaries (e.g. Hadamard).}
 $$
\end{nonamethm}

\subsubsection*{Main Conceptual Contributions}

\begin{itemize}
    \item We provide a unified framework from which to consider models of quantum computation that involve free operations acting on some fixed resourceful state.
    \item We give evidence that any model of quantum computation must involve the resource of coherence in the operations (exemplified by the Hadamard gate). That is, coherence cannot be siphoned off to some supplementary state, unlike in the cases of magic in magic state injection or entanglement in measurement based quantum computation.
\end{itemize}

\subsubsection*{Main Technical Contributions}

Recall that the dephasing map $\Delta$ sets all off-diagonal terms of the density matrix in the computational basis to zero, and the trace distance is defined as $D(\rho, \sigma) = \frac{1}{2} \norm{\rho - \sigma}_1$.

\begin{nonamethm}
    \textbf{\cref{lem:exact} (informal).} 
    If a channel $\mathcal{E}$ commutes with the dephasing map $\Delta$, then for any state $\ancillaket$ the induced channel $\rho \mapsto \mathcal{E}(\rho \otimes \ancillaproj)$ cannot implement any coherent unitary.
\end{nonamethm}

This generalises an observation made in an erratum to \cite{dana2017resource} that shows a similar result for qubits. 

Our second main technical result is a robust extension of this to the approximate case, when specifically considering $n$ Hadamard gates, of particular relevance in quantum computation.

\begin{nonamethm}
    \textbf{\cref{lem:0ton_approx} (informal).} 
    Let $\mathcal{E}$ be a channel that commutes with the dephasing map $\Delta$, let $H$ denote the Hadamard gate and $D$ denote trace distance. Then for any state $\ancillaket$ we have 
$$
\max_{\rho} ~ D \bigg (\mathcal{E} (\rho \otimes \ancillaproj) ~,~ H^{\otimes n} \rho H^{\otimes n} \bigg) \geq 1 - \frac{1}{2^n}.
$$
\end{nonamethm}

Thirdly, we show that $k$ Hadamards, incoherent unitaries, classical control and an arbitrary ancilla cannot be used to implement $n > k$ Hadamards exactly and deterministically.

\begin{nonamethm}
    
\textbf{\cref{lem:coherencerank} (informal). }
Let $U=U_k V_k \dots U_1 V_1 U_0$ be a product of unitaries, comprised of $k$ Hadamards $V_i$ and incoherent $U_i$. Then for any state $\ancillaket$ we have
$$
\text{Tr}_2\left ( U \rho \otimes \ancillaproj U^\dagger \right ) = H^{\otimes n} \rho H^{\otimes n} \qquad \forall \rho, \hspace{40pt} \implies \quad n \leq k.
$$
\end{nonamethm}
Whilst these results stand independently in the study of coherence, we interpret them in quantum computation by showing that certain operationally motivated channels satisfy the conditions. 

Supporting results include showing that quantum-controlled incoherent unitaries are incoherent (\cref{lem:controlled_incoherent}), placing bounds on the coherence rank of a state after $k$ Hadamards have been applied (\cref{lem:cohrankprod}), and proving that if the marginal of a unitary channel acting on an input state and fixed state is unitary, then the other marginal must be independent of the input state (\cref{lem:ancilla_indep}).

We now provide further background and definitions, before more formally introducing our framework in the next section.

\subsection{Background}

The seminal result of the Gottesman-Knill theorem  \cite{gottesman1997stabilizer, aaronson2004improved, jozsa2013classical, nielsen2002quantum} states that any quantum computation consisting of Clifford operations (comprised of CNOT, Hadamard and phase gates), can be simulated efficiently on a classical computer. It is known that the $T$ gate elevates this set to universality, and the Clifford + $T$ gate set is perhaps the most widely considered universal set of gates. Motivated by error-correction and fault-tolerance considerations \cite{campbell2017roads}, in place of directly applying a $T$ gate, one can perform adaptive Clifford operations on an arbitrary input state and a so-called \textit{magic state}, to implement the $T$ gate deterministically. This is often referred to as a \textit{gadget}, where one replaces all uses of a given gate with this subroutine, consuming a resourceful state in the process. See \cref{ex:msi} below for further detail here.

Another example of a gadget-based approach can be found in recent work on matchgate circuits \cite{hebenstreit2019all, hebenstreit2020computational}. Matchgates are a family of two-qubit gates, inspired by fermionic systems, that can be written as the direct sum of two single qubit gates with the same determinant, acting respectively in the even and odd parity subspaces \cite{jozsa2008matchgates}. It is known that circuits composed of matchgates acting only on nearest-neighbour qubits are classically simulable, however any family of quantum circuits can be simulated efficiently with circuits composed of matchgates acting on next-nearest-neighbour qubits \cite{terhal2002classical, jozsa2008matchgates}. Hence nearest-neighbour matchgates can be augmented to universality using SWAP gates, analogously to Clifford circuits and $T$ gates. The work of \cite{hebenstreit2019all, hebenstreit2020computational} highlights this connection (see Figure 1 in \cite{hebenstreit2020computational}), and shows the existence of a SWAP state, which can be consumed under adaptive nearest-neighbour matchgates to implement the SWAP gate. This provides a parallel gadget based approach to the Clifford + $T$ case, in which resourceful states are consumed to implement resourceful gates, enabling universality.

Measurement based-quantum computing (MBQC) \cite{jozsa2006introduction} generally refers to any model of quantum computation in which the primary allowed operations are measurements. The foremost example of this is the so-called \textit{one-way} MBQC model \cite{raussendorf2001one, nielsen2006cluster}, in which adaptive single qubit measurements are performed on some fixed resource state. This is usually taken to be a cluster state, a state in which qubits are laid out in a rectangular grid, initialised to $\ket{+}$ states, and controlled-$Z$ gates are applied between neighbouring qubits. Another model is teleportation-based quantum computation, which proceeds by using Bell measurements to teleport gates \cite{gottesman1999quantum, jozsa2006introduction}. 

The above examples are all connected: they all relate to performing some perceived free operations on an apparently resourceful fixed state. In the magic state injection model, one may consider Clifford operations as free, and the resource state contains the \textit{magic} needed for the computation. In the standard MBQC framework, local measurements are considered free (one can generalise this to consider arbitrary \textit{local operations and classical communication} (LOCC) operations \cite{van2007fundamentals}), and the resource state contains all the \textit{entanglement} needed for the computation.

The framework of \textit{quantum resource theories} \cite{chitambar2019quantum} aims to identify components of quantum theory that are non-classical in some sense, by defining so called \textit{free} sets of states, and allowed channels and measurements. One can then define resource quantifiers, such as the distance a given object is away from the free set, or finding a minimal convex combination of an object and free object. This paradigm has roots motivated by thermodynamics, and the archetypal quantum resource theory is that of entanglement. Here the free states and allowed channels can respectively taken to be separable states and LOCC. The resource theory of coherence has also gathered a lot of attention in recent years \cite{streltsov2017colloquium}, and is highly relevant to this work. In this context, the set of free states are those which are diagonal in some fixed basis (termed \textit{incoherent}), however there are multiple approaches to defining the allowed class of operations, which has lead to fruitful and nuanced discussion \cite{chitamber2016critical} -- see \cref{app:resourcecoherence} for further more on this.

When considering the computational power of a set of quantum operations, there are multiple approaches one can take. One can consider \textit{classical simulability}, namely if one can efficiently perform the same calculation on a classical computer. Here there are several subtleties: how to precisely quantify `efficiently', and the exact simulation task considered; for example the ability to sample from measuring the final state in the computational basis (weak simulation), or the ability to compute or bound a given output probability of the final state (strong simulation) - see e.g. \cite{ van2007classical, nest2008classical, xu2023herculean}. Another angle is to consider \textit{universality}, that is, the ability of the operations to implement any unitary or prepare any quantum state, with extensions including notions of approximate and probabilistic universality \cite{mora2010universal, van2007fundamentals}. These ideas are not independent: for quantum computers to be strictly more powerful than classical computers, one would expect that efficient classical simulation of a universal quantum device is not possible, however the inability to classically simulate a quantum process efficiently does in general not imply universality (for example, consider approaches to so-called `quantum computational supremacy' \cite{harrow2017quantum}).

In this work, we focus on the notion of universality. We consider the resource of coherence in gadget-based approaches to quantum computation through studying the role of the Hadamard gate. Specifically, we ask whether given access to incoherent unitaries (i.e. unitaries unable to generate superpositions when acting on computational basis states), computational basis measurements, and classical control (e.g. applying unitaries conditioned on previous measurement outcomes) if there exists a quantum state (which can be completely arbitrary) such that one can implement Hadamard gates, either exactly or approximately. To phrase this in a slightly contrived fashion and give broader motivation, suppose some distant civilisation are capable of preparing and transporting some complicated resourceful state. What are the minimal operations that are necessary for the recipient in order for them to be able to perform universal quantum computation? In this work, we will provide evidence of where this resource `cut' lies: the ability to perform coherent operations (or incompatible measurements) are all that is necessary, everything else can be moved into the resource state. Complementary results also show that the ability to perform the Hadamard gate is sufficient in this context \cite{takeuchi2019quantum}, hence we draw closer to a complete answer to this question. Along the way, we show several results that may be of broader interest in quantum information, computation, and resource theories.

We will now introduce some examples that explain the above areas in more detail, providing concrete motivation and serving as a reference for the rest of the document.

\begin{example}
\label{ex:msi}
 Consider the gate set of Clifford + $T$, comprised of gates from $\{CNOT, H, S, T\}$, where $H$ is the Hadamard gate and $S$ is the phase gate. As discussed above, the $T$ gates may be implemented by performing adaptive Clifford operations on supplementary $T$ states. A natural question is: where else could we put the `cut' between gates and states? Could it be possible to do universal quantum computation with only adaptive CNOT gates acting on some supplementary resourceful state? 

To provide a more concrete basis for this question, consider the following circuit, valid for all diagonal gates $U = \begin{pmatrix} 1 & 0 \\ 0 & e^{i\theta} \end{pmatrix}$ and qubit input $\ket{\psi} \in \mathbbm{C}^2$.
\[
  \Qcircuit @C=1em @R=2em {
\ket{\psi} &&& \ctrl{1} & \qw & \gate{U^2} & \qw && U\ket{\psi} \\
U\ket{+} &&&  \targ & \qw &\meter \cwx[-1] 
}
\]
    Magic state injection is the special case of this when $U=T$:
    
    \[
  \Qcircuit @C=1em @R=2em {
\ket{\psi} &&&& \ctrl{1} & \qw & \gate{S} & \qw && T\ket{\psi} \\
\ket{T} := T \ket{+} &&&&  \targ & \qw &\meter \cwx[-1] 
}
\]
Observe that for $U=Z$, this becomes
    
    \[
  \Qcircuit @C=1em @R=2em {
\ket{\psi} &&& \ctrl{1} & \qw  & \qw & Z\ket{\psi} \\
\ket{-} &&&  \targ & \qw &\meter
}
\]
as $Z^2 = \id$. Hence given access to CNOTs and $\ket{-}$ states we can implement the $Z$ gate deterministically. We can apply this as a subroutine, enabling us to implement the phase gate $S$ using two CNOTs and the state $\ket{-}\bigg( \frac{\ket{0} + i \ket{1}}{\sqrt{2}} \bigg )$. Iterating in this way, we can reach any gate of the form $U_k = \begin{pmatrix} 1 & 0 \\ 0 & e^{\frac{2 \pi i}{2^k}} \end{pmatrix}$ for $k \in \mathbbm{N}$ (note that $U_2 = T$). Hence for Clifford + $T$, we can replace the $S$ and $T$ gates with gadgets, and perform universal quantum computation with ability to only perform CNOTs and Hadamard on some supplementary state. However it is not clear how to restrict this gate set further when only computational basis measurements are permitted.
\end{example}

\begin{example} \label{ex:hadgadget}
    So-called Hadamard gadgets are known to exist \cite{heyfron2018efficient, Beaudrap2020FastAE, bremner2011classical}, where they play roles relating to compilation and simulation of quantum circuits. For example, the following circuits appear in \cite{heyfron2018efficient} and \cite{Beaudrap2020FastAE} respectively:
\[ \label{eq:hadgadget1}
  \Qcircuit @C=1em @R=1.5em {
\ket{\psi} &&\gate{S} & \ctrl{1} & \qw & \targ & \ctrl{1} &  \gate{X} & \qw &&  H\ket{\psi} \\
\ket{+}  &&\gate{S} & \targ & \gate{S^\dagger} &  \ctrl{-1} & \targ & \meterB{ \text{\tiny $X$}  } \cwx[-1] 
} 
\]

\[ \hspace{50pt}  \label{eq:hadgadget2}
  \Qcircuit @C=1em @R=1.5em {
\ket{+} && \ctrl{1} & \qw  &  \gate{X} & \qw &&  H\ket{\psi} \\
\ket{\psi}    & & \control  \qw & \qw  &  \meterB{  \text{\tiny $X$}  } \cwx[-1] 
} 
\]
However, they crucially rely on $X$ basis measurements, that is, measurements in the coherent basis $\{\ket{+}, \ket{-}\}$. In this work, we will show that such gadgets cannot exist if one restricts to computational basis measurements.

\end{example}

\begin{example}

In \cite{takeuchi2019quantum} it is shown that measurement-based quantum computing is possible with adaptive $X$ and $Z$ measurements alone. This can alternatively can be viewed as the ability to only perform the Hadamard gate and measure in the computational basis. It is clear that the resource state here cannot be a graph state, as graph states are stabiliser states, and as Hadamard is a Clifford gate we could simulate the whole computation using the Gottesman-Knill theorem. Indeed the state considered in \cite{takeuchi2019quantum} is a hypergraph state \cite{rossi2013quantum}, formed by initialising all qubits to $\ket{+}$ and performing multiply controlled $Z$ gates for each hyperedge. In particular, one can see that as $CCZ$ is not a Clifford operation, hypergraph states will not be stabiliser states in general.

This example shows that given the ability to only perform adaptive Hadamard gates and computational basis measurements, there exists a resourceful ancillary state such that universal quantum computation is possible.
\end{example}

\begin{example}
\label{ex:erratum}
\textit{Incoherent operations} (IO) are defined as channels admitting a Kraus decomposition $\mathcal{E}(\rho) = \sum_\alpha K_\alpha \rho K_\alpha^\dagger$ such that $K_\alpha \rho K_\alpha^\dagger$ is an incoherent state for each $\alpha$ and each incoherent state $\rho$.
 It is known that these channels supplemented with a maximally coherent state $\ket{\Psi_d} = d^{-\frac{1}{2}}\sum_{k=0}^{d-1} \ket{k} \in \mathbbm{C}^d$ are able to implement any quantum channel \cite{dana2017resource}. 
 
 It was originally claimed in the same paper that a similar result, namely the ability to implement any unitary given access to $\ket{\Psi_d}$, held for \textit{strictly incoherent operations} (SIO), which are IO with the additional property that $K_\alpha^\dagger \rho K_\alpha$ is an incoherent state for each $\alpha$ and each incoherent state $\rho$. However, it was later shown in an erratum to \cite{dana2017resource} that their proof was invalid as the operations used were not SIO. In this erratum, the authors gave a simple argument that if a qubit channel commutes with the dephasing map (which all SIO do), then even supplemented with an arbitrary ancilla one cannot implement any coherent unitary.

 This example highlights an interesting distinction: IO can `unlock' the resource in a supplementary state, whereas the slightly weaker class of SIO are unable to access any of this state resource. In this work we show that a class of operations motivated by quantum computation gadgets are also unable to harness the power in a supplementary coherent state. See \cref{app:resourcecoherence} for further background on the resource theory of coherence, and how our work relates to this topic.
\end{example}

\begin{example}
    In \cite{shi2002both}, it is shown that CNOT and any single qubit gate whose square is basis changing (i.e. coherent) is universal for quantum computation. The same result is also shown for the Toffoli gate and any single qubit basis changing gate. A simpler proof for the case of Toffoli + Hadamard was presented in \cite{aharonov2003simple}. These results are conceptually fascinating as the Toffoli gate is universal for classical computing, so by including the `quintessentially quantum' Hadamard gate one elevates classical universality to quantum universality. As the above gates are real, one uses an additional ancilla to simulate complex numbers.
\end{example}

The above examples motivate the following questions:
\begin{enumerate}[(1)]
\item Is it possible to provide a gadget for the Hadamard gate using only incoherent unitaries, computational basis measurements, and an ancilla?
\item Is universal quantum computation possible with only incoherent unitaries acting on some resourceful state? Or does any universal model require some coherence (e.g. Hadamard) in the operations?
\item Where can we put the `cut' between states and operations for quantum computation in general?
\item If coherence must be present in the operations, how much coherence is necessary and sufficient for universality?
\item Is there a connection between the role of coherence in gadget-based approaches, and the role of measurements in MBQC approaches?
\end{enumerate}

The purpose of this paper is to initiate this line of research, and make progress in answering some of these questions. We provide answers in the negative to points (1) and (2), whilst discussing (3) - (5) towards the end of the document and motivating them for future research.

In particular, we rule out the existence of circuits of the following general form, for $U$ and $V$ incoherent unitaries (e.g. products of CNOTs and $T$ gates):

\[
  \Qcircuit @C=1em @R=1em {
\ket{\psi} \quad &&\qw& \multigate{3}{\quad U \quad} & \qw & \gate{V} & \qw && H\ket{\psi} \\
 &&\qw & \ghost{\quad U \quad} & \qw &\meter \cwx[-1] \\ 
 && \qvdots & & &  \qvdots \\ 
 &&\qw & \ghost{\quad U \quad} & \qw &\meter  \inputgroupv{2}{4}{1em}{2em}{\ancillaket} \\
}
\]
Note that we know that the above diagram is possible with $\ancillaket = \ket{+}$ if we instead allow $X$ measurements, as shown in \cref{ex:hadgadget}.

At this stage, one might worry that the ancillary state $\ancillaket$ cannot be useful only as the set of unitaries considered are completely resourceless. A priori, it is possible that the use of a single Hadamard gate could allow incoherent unitaries to unlock all the power from the ancillary state to implement a more coherent gate, for example, two Hadamards. Hence we also consider the natural extension of whether incoherent unitaries, computational basis measurements, an ancilla, and $k$ Hadamards can simulate $n > k$ Hadamards. In the case of $k=1$ and $n=2$, the corresponding diagram could be of the form:

\[
  \Qcircuit @C=1em @R=1em {
 &&\qw& \multigate{4}{\quad U_1 \quad} & \gate{H} & \multigate{4}{\quad U_2 \quad} &\multigate{1}{ ~ V ~ } & \qw &&  \\
 &&\qw & \ghost{\quad U_1 \quad} & \qw &  \ghost{\quad U_2 \quad} & \ghost{~V~} & \qw\inputgroupv{1}{2}{0em}{1.3em}{\ket{\psi}}\gategroup{1}{8}{2}{8}{.8em}{\}} &&&  \raisebox{2em}{\ensuremath{H_1 H_2 \ket{\psi}}} \\ 
 &&\qw & \ghost{\quad U_1 \quad} & \qw &  \ghost{\quad U_2 \quad} &\meter \cwx[-1] \\ 
 && \qvdots & && &  \qvdots \\ 
 &&\qw & \ghost{\quad U_1 \quad} & \qw &  \ghost{\quad U_2 \quad} & \meter  \inputgroupv{3}{5}{0em}{1.5em}{\ancillaket} \\
}
\]

where $U_i$ and $V$ are incoherent unitaries and $H_i$ denotes a Hadamard gate on the $i$-th qubit. We are also able to rule out this case in this work. This demonstrates the importance of having the ability to generate large amounts of coherence in any model of quantum computation, directly contrasting with the magic state injection case in which all the `non-stabiliserness' can be placed in supplementary ancillas.

The document is organised as follows. After fixing notation, we motivate a general framework for quantum computation involving some free unitaries acting on a resourceful state. We then apply this to coherence in our results section, focusing on the case of incoherent resources attempting to use a supplementary state to implement $H^{\otimes n}$ (we refer to this as the case $0 \mapsto n$), as well as the case of incoherent resources and the use of $k$ Hadamards to implement $H^{\otimes n}$ (the case $k\mapsto n$). We consider the cases of exact, deterministic, approximate and probabilistic implementation. We conclude with a discussion of the key concepts our work relates to, and provide several novel research problems as outlook. \cref{app:resourcecoherence} provides further background to resource theories and coherence.

\subsection{Notation and Definitions}\label{subsec:defns}
Let $\mathcal{L}(\mathcal{H})$ denote the set of linear maps on a Hilbert space $\mathcal{H}$. In this work all Hilbert spaces will be finite dimensional, $\mathcal{H} \cong \mathbbm{C}^d$, and in particular we will focus on qubit systems. Quantum states are positive semi-definite elements of $\mathcal{L}(\mathcal{H})$ with unit trace: we denote this set by $S(\mathcal{H})$. Quantum channels are completely positive trace-preserving (CPTP) maps $\mathcal{E}: \mathcal{L}(\mathcal{H}_1) \rightarrow  \mathcal{L}(\mathcal{H}_2) $. Similarly, quantum subchannels are completely positive trace non-increasing maps from $ \mathcal{L}(\mathcal{H}_1)$ to $\mathcal{L}(\mathcal{H}_2) $. We may abuse notation by referring to a unitary channel $\mathcal{U}(\cdot) = U (\cdot) U^\dagger$ simply as $U$, and by writing $\ket{\psi^n}$ in place of $\ket{\psi}^{\otimes n}$.

The Pauli matrices are
\[
X= \begin{pmatrix}
    0 & 1 \\
    1 & 0
\end{pmatrix},
\qquad Y= \begin{pmatrix}
    0 & -i \\
    i & 0
\end{pmatrix},
\qquad Z= \begin{pmatrix}
    1 & 0 \\
    0 & -1
\end{pmatrix},
\]
and the $n$-qubit Pauli group $\mathcal{P}_n$ is generated by tensor products of Pauli matrices, elements being referred to simply as `Paulis'. The Clifford group is the normaliser of $\mathcal{P}_n$ and is generated by (tensor products of) the following gates
\begin{align}
&\text{\underline{Hadamard}} &&\text{\underline{Phase}} &&\text{\underline{Controlled-NOT}} \\\nonumber
&H= \frac{1}{\sqrt{2}}\begin{pmatrix}
    1 & 1 \\
    1 & -1
\end{pmatrix},
\qquad &&S= \begin{pmatrix}
    1 & 0 \\
    0 & i
\end{pmatrix},
\qquad &&CNOT= \begin{pmatrix}
    1 & 0 & 0 & 0 \\
    0 & 1 & 0 & 0 \\
    0 & 0 & 0 & 1 \\
    0 & 0 & 1 & 0
\end{pmatrix}.
\end{align}

The $T$ gate and the $T$ state are respectively defined as
\[
T= \begin{pmatrix}
    1 & 0 \\
    0 & e^{\frac{i \pi}{4}}
\end{pmatrix}, \hspace{30pt} \ket{T} = \frac{1}{\sqrt{2}}\left ( \ket{0} + e^{\frac{i \pi}{4}} \ket{1} \right ).
\]

The trace distance on quantum states is defined as

\[
D(\rho, \sigma) := \frac{1}{2}\norm{\rho-\sigma}_1, \label{eq:tracedistance}
\]
where $\norm{M}_1 = \text{Tr}(\sqrt{M^\dagger M})$. The trace distance has the following properties for all states $\rho$, $\sigma$: (i) positivity: $D(\rho, \sigma) \geq$ 0 with equality $\iff \rho=\sigma$ (ii) symmetry: $D(\rho, \sigma) = D(\sigma, \rho)$ (iii) triangle inequality: $D(\rho, \sigma) \leq D(\rho, \omega) + D(\omega, \sigma)$, (iv) contractivity: $D(\Lambda (\rho), \Lambda (\sigma)) \leq D(\rho,  \sigma) $ for all quantum channels $\Lambda$.

The induced trace distance on channels results from maximising over possible input states:

\[\label{eq:induced_trace_dist}
\mathcal{D}(\mathcal{E}, \mathcal{V}) := \max_{\rho} ~  D\bigg (\mathcal{E}(\rho),\mathcal{V}(\rho) \bigg ).
\]
We say that a channel $\mathcal{E}$ \textit{$\epsilon$-approximates} a channel $\mathcal{V}$ if they they are at most $\epsilon$ close in this induced trace norm.

In the resource theory of coherence, one fixes a basis $\{ \ket{x} \}$ of $\mathbbm{C}^d$  (which we may refer to as the computational basis). A state is then \textit{incoherent} if it can be written as
\[
\rho = \sum_x p_x \ketbra{x}
\]
in this basis. We denote the set of incoherent states (for an implied fixed dimension) by $\mathcal{I}$ -- note that it is convex and compact. A unitary $U$ is \textit{incoherent} relative to the basis $\{ \ket{x} \}_{x=1}^d$  if it can be written as
\[
U = \sum_{x=1}^d e^{i \theta_x} \ketbra{\pi(x)}{x}
\]
for $d$ real numbers $\theta_x$ and some permutation $\pi$ on $d$ elements. In particular, incoherent unitaries map a computational basis state to another computational basis state, possibly multiplied by some phase. They are precisely the maximal set of unitaries mapping $\mathcal{I}$ to itself. Examples of incoherent unitaries include the Pauli operators, the phase and $T$ gates, CNOT, SWAP, and the Toffoli gate. Examples of unitaries that are coherent (i.e. not incoherent, able to generate coherence) include the Hadamard gate, the Fourier transform, and $X$ rotations $e^{i\theta X}$ for $\theta \notin \{n \pi ~ ; ~ n \in \mathbbm{Z} \}$. Coherent unitaries may also be called \textit{basis changing} \cite{shi2002both}.

For a fixed basis $\{ \ket{x} \}$, the \textit{dephasing map} is defined as 
\[
\Delta (\rho) := \sum_x \ketbra{x} \rho \ketbra{x}. \label{eq:dephasing}
\]
This has the effect of removing the off-diagonal elements on a density operator, hence $\Delta (\rho) \in \mathcal{I}$ for all states $\rho$. We may use the symbol $\Delta$ multiple times in an expression even though they may act on quantum states of different dimensions, which can be inferred from context. We use the term \textit{incoherent resources} informally to refer to operations arising from incoherent unitaries, classical control, computational basis measurements and preparation, and partial traces. Throughout we will use $\ancillaket$ as notation for a pure ancilla, and $\ancillamixed$ as notation for a mixed state ancilla.

See \cref{subsec:prelim} for more preliminaries, and \cref{app:resourcecoherence} for further background on resource theories, in particular that of coherence.

\section{Framework} \label{sec:framework}
In this section we will consider a general paradigm for quantum computation using some additional ancillary state as a resource. Consider the following: \\

\boxed{\begin{minipage}{.5\linewidth}
\textbf{Free operations: }
\begin{itemize}
 \setlength\itemsep{-1em}
    \item Preparation of computational basis states.\\
    \item Measurement in the computational basis. \\
    \item Classical control and adaptivity. \\
    \item Some set of unitaries $\mathcal{U}$.
\end{itemize}
\end{minipage}}\hspace{20pt}
\textbf{\huge $\bm{+}$}
\begin{minipage}{.35\linewidth}
\vspace{-20pt}
\begin{align}
&\text{\hspace{-10pt}an additional}\label{eq:operations}\\\nonumber
&\text{\hspace{-10pt}resourceful state $\ancillaket$.}\end{align} 
\end{minipage} \\[8pt]

Here `classical control and adaptivity' refers to the ability to perform a unitary from $\mathcal{U}$ or measurement classically conditioned on the outcomes of previous measurements.

\subsection{Examples}

Many approaches to quantum computation fall into the above framework -- see \cref{tab:examples} for a list of examples. In the standard circuit model, we take the set of unitaries $\mathcal{U}$ to be a universal gate set, and do not consider a supplementary state $\ancillaket$ (i.e. it is redundant here). In the magic state injection model, the set $\mathcal{U}$ is taken as Clifford gates, and the supplementary state $\ancillaket$ can be taken as a tensor product of $T$ states $\ancillaket = \ket{T}^{\otimes{m}}$, where $m$ would be the number of $T$ gates in the desired circuit. For efficient quantum computation, the depth of a family of circuits should grow at most polynomially in terms of the number of qubits $n$, hence in practice we would require $m=O(\text{poly}(n))$. Similarly, we could also take $\mathcal{U}$ to be nearest-neighbour matchgates, and $\ancillaket$ to be a polynomial number of SWAP states (as defined and discussed in \cite{hebenstreit2019all}).

\renewcommand{\arraystretch}{3}
\begin{table}[htbp!]
    \centering
\begin{tabular}{|c||c|c|c|}\hline
    Model & Operations & State & References \\\hline
     Circuit & Universal gate set & $-$ & \cite{nielsen2002quantum} \\
     Magic State Injection & Clifford & $\ket{T}^{\otimes p}$ &\cite{bravyi2005universal} \\
     Matchgates & \shortstack{Nearest neighbour \\ matchgates} & $\ket{SWAP}^{\otimes p}$ & \cite{hebenstreit2019all, hebenstreit2020computational} \\
     1-way MBQC & LOCC & Graph state & \cite{briegel2009measurement, van2007fundamentals}\\
     1-way MBQC & Hadamard & Hypergraph state & \cite{takeuchi2019quantum} \\
     Teleportation MBQC & \shortstack{Rotated Bell unitaries \\ Paulis} & $\ket{\phi^+}^{\otimes p}$& \cite{gottesman1999quantum, jozsa2006introduction} \\\hline
\end{tabular}
    \caption{Comparison of different models of quantum computation that fall into the framework summarised in \cref{eq:operations}, we are also allowing classical control and computational basis measurement and preparation freely. $\ket{T}$, $\ket{SWAP}$ and $\ket{\phi^+}$ respectively refer to the $T$ state, the SWAP state \cite{hebenstreit2019all}, and the maximally entangled state. Here the tensor power $p$ should be taken as some polynomial of the number of qubits. In place of considering measurements, we can consider the corresponding unitaries that rotate the computational basis into the appropriate basis. See also Table 1 in \cite{takeuchi2019quantum} for a more extensive summary of MBQC approaches using different measurement bases.}
    \label{tab:examples}
\end{table}
\renewcommand{\arraystretch}{1}

For measurement-based quantum computation, we take $\ancillaket$ to be some entangled state, such as a graph or hypergraph state. The operations permitted here are usually taken to be local projective measurements, but we can include them in the above framework in the following way. Instead of measuring in a specific basis, we could first apply a local unitary and then measure in the computational basis. Explicitly, if we wish to measure observable $\mathcal{P} = \sum_x \alpha_x \ketbra{\psi_x}$, we could instead perform the unitary $U=\sum_x \ketbra{x}{\psi_x}$ and measure in the computational basis to the same effect\footnote{This is effectively the Heisenberg picture.}. Hence we can incorporate measurement-based approaches here, however note that the reverse direction does not hold: the ability to perform measurements in various bases does not directly imply the ability to perform the corresponding unitaries\footnote{However as we have seen, there exist Hadamard gadgets if one is allowed to perform an $X$ measurement. Hence with respect to incoherent unitaries, the Hadamard gate and $X$ measurement are equivalent in some sense.}. We also remark that if the set of unitaries $\mathcal{U}$ are single qubit unitaries, then clearly we are in the measurement based scenario (as opposed to some gate injection scheme).

Let us also recall teleportation-based MBQC \cite{gottesman1999quantum, jozsa2006introduction}. Consider the following two circuit identities:

\[
  \Qcircuit @C=1em @R=1em {
\ket{\psi} && \multigate{1}{B} & \meter & \control \cw \cwx[1] \\
 &&  \ghost{B} & \meter & \control \cw \cwx[1]  \\
 &&  \qw & \qw & \gate{V} & \qw  \hspace{30pt} U\ket{\psi}, \inputgroupv{2}{3}{0.1em}{1.1em}{\id \otimes U\ket{\phi^+} \hspace{30pt}}} \hspace{90pt}
 \Qcircuit @C=1em @R=1em {
\ket{\psi} && \multigate{1}{B(U)} & \meter & \control \cw \cwx[1] \\
&&  \ghost{B(U)} & \meter & \control \cw \cwx[1]  \\
 &&  \qw & \qw & \gate{V} & \qw \hspace{30pt} U\ket{\psi}, \inputgroupv{2}{3}{0.1em}{1.1em}{\ket{\phi^+} }}
\]

where $B:=(H \otimes \id) CNOT_{12}$ denotes the Bell unitary, $B(U):= (\id \otimes U ) B$ is the rotated Bell unitary, $\ket{\phi^+} = d^{-\frac{1}{2}}\sum_{k=0}^{d-1} \ket{k k} \in \mathbbm{C}^{d^2}$ is the maximally entangled state, and $V$ is a Pauli correction term. Hence given a pool of Bell states or rotated Bell states we can achieve universality in this way. One can also teleport the CNOT gate in a similar fashion using a 4 qubit Bell state.

Pauli based computation (PBC) \cite{bravyi2016trading} proceeds by adaptively performing non-destructive Pauli measurements on $\ket{T}$ states as input. To incorporate this into our framework by phrasing it in the language of unitaries and computational basis measurements, we would have to find unitaries $U$ such that $U C U^\dagger = P$, where $C$ is a non-degenerate Hermitian operator diagonal in the computational basis, and $P$ is a tensor product of Pauli operators. Note that measuring the operator $Z \otimes \dots \otimes Z$ is not equivalent to measuring in the computational basis (it has $2 $ outcomes as opposed to $2^n$).

We also remark that to consider the notion of \textit{efficient} universal quantum computation, it is necessary to consider a family of sets of unitaries on $n$ qubits, and the size of ancillary state $\ket{\ancillapure(n)}$ should scale at most polynomially with $n$ \cite{van2007fundamentals, mora2010universal}. Our results allow the ancilla to be of arbitrary size, and we show the impossibility of providing a Hadamard gadget (using incoherent resources) within this framework. To extend our discussion to the MBQC framework, the scaling size of the ancillary state must be taken into account. To see this, recall that an $\epsilon$-net is a set of states such that any state is within distance $\epsilon$ of some state in the net. One could take the tensor product all the states in such a net as the ancilla. Then for any given state, there would exist a marginal of $\ancillaket$ within distance $\epsilon$. Thus it may appear that this would lead to a universal model of quantum computation in which the only operations required are partial traces. However, such a state would not scale polynomially in the number of qubits -- we elaborate on this concept in \cref{subsec:discussmbqc}.

\subsection{General Form of Operations}

In order to make concrete statements, we now motivate an expression for a general operation within the above framework. We will first need a short definition:
\begin{definition} \label{def:controlledU}
Given some set of unitaries $\mathcal{U}$ and a preferred basis $\{ \ket{x} \}$, we denote by $\mathcal{C}(\mathcal{U})$ the corresponding set of generalised controlled unitaries. 
In particular, on $n$ qubits these are of the form
\[
\sum_{x \in S} \ketbra{x}{x} \otimes U + \sum_{y \in S^c}\ketbra{y}{y} \otimes \mathbbm{1},
\]
where $U \in \mathcal{U}$ acts on $k \leq n$ qubits, $S \subseteq \{0,1\}^{n-k}$ and $S^c$ is the complement of $S$ in $\{0,1\}^{n-k}$.
\end{definition}

This definition simply describes the quantum equivalent of classically controlled operations: given a computational basis vector, a unitary is performed on a subset of the qubits only for some specific values of the remaining bits. Observe that CNOT, controlled-$Z$ and Toffoli fall under this definition, and general controlled operations are also discussed in \cite{nielsen2002quantum, barenco1995elementary} (but for the case where $S$ contains a single bitstring). Note that this definition also encompasses the case of $ \id \otimes U$ (here $S=\{0,1 \}^{n-k}$ and $S^c$ is the empty set), and also the case of $\id \otimes \id$; we will use the term `controlled-$U$' in this broader sense.

Now consider the operations arising from \cref{eq:operations}. We can without loss of generality append all computational basis states at the beginning, and absorb them into $\ancillaproj$. Hence we can consider the input to be $ \rho\otimes \ancillaproj $.

One could then apply a sequence of intermittent unitaries and computational basis measurements, which could be an adaptive process conditioned on some classical information and the outcomes of previous measurements. Note that we can delay these measurements to the end, by instead using unitaries from the controlled set $\mathcal{C}(\mathcal{U})$. That is, if a unitary $U$ is to be applied on system $A$ conditioned on the outcome of some previous measurement on system $B$, we could instead apply a unitary to system $B$ (mapping the original measurement bases to the computational basis), and perform a controlled-$U$ operation on system $A$ with system $B$ as control. We can then defer the measurement of system $B$ until the end of the computation. This is often referred to as the \textit{principle of deferred measurement} \cite{nielsen2002quantum}, see \cref{fig:deferred_measurement}. For example in one-way MBQC, this would result in a circuit of controlled single qubit unitaries being applied to the cluster state. Finally, one could disregard some of the systems, corresponding to a partial trace. Put together, this now leads to the following observation.

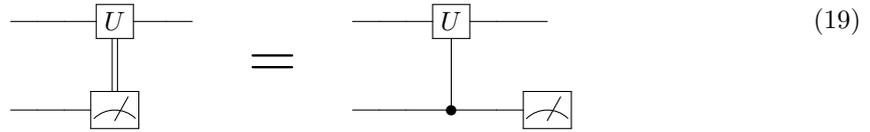
\begin{figure}[htbp!]
    \centering
    \[
  \Qcircuit @C=1em @R=2em {
&\qw &\qw & \gate{U} & \qw &\qw  &&&\raisebox{-4em}{\Huge =}&&&& \qw &\qw & \gate{U} & \qw &\qw   \\
 &\qw &\qw &\meter \cwx[-1] && &&&&&&&  \qw &\qw &\ctrl{-1]} & \qw & \meter 
} 
\]
    \caption{Principle of Deferred Measurement \cite{nielsen2002quantum}.}
    \label{fig:deferred_measurement}
\end{figure}

\begin{observation} \label{ob:map}

The most general channel possible to implement within the above framework is given by

\[
    \mathcal{E}(\rho) = \text{Tr}_X \bigg (  U ( \rho \otimes \ancillamixed ) U^\dagger \bigg ). \label{eq:general_channel}
\]
Here $U$ belongs to the set of controlled unitaries $\mathcal{C}(\mathcal{U})$, $\text{Tr}_X$ denotes a partial trace on some of the subsystems, and $\ancillamixed$ is an arbitrary fixed state.

The most general probabilistic (i.e. trace non-increasing) operation possible to implement within the above framework is given by a convex combination of operations of the form

\[
    \mathcal{E}^x(\rho) = \text{Tr}_X \bigg (  ( \id \otimes \ketbra{x} ) ~ U ( \rho \otimes \ancillamixed )  U^\dagger \bigg ), \label{eq:general_subchannel}
\]
where the projector $\ketbra{x}$ denotes a measurement in the computational basis on some of the subsystems.

\end{observation}

Note that the channel in \cref{eq:general_channel} could include some final computational basis measurements, but as we consider the overall channel to be independent of the outcomes of these measurements this corresponds to a partial trace. See \cref{rem:pio_channels} in \cref{app:resourcecoherence} for a brief comment on how the channels we consider relate to the resource theory of coherence.

In the probabilistic case, we will be interested in the case where these subchannels are proportional to a unitary channel. Note that in general for a subchannel $\mathcal{E}$ to be proportional to a channel $\mathcal{V}$, i.e.
\[
\frac{\mathcal{E} (\rho)}{\text{Tr} ( \mathcal{E} (\rho))} = \mathcal{V}(\rho),
\]
we must have that $\text{Tr} ( \mathcal{E}^x (\rho))$ is independent of $\rho$, to ensure linearity.

    \begin{remark}
We can take the ancilla to be pure without loss of generality, as we can always purify the state. That is, given $\ancillamixed=\sum p_n \ketbra{n}$ in spectral decomposition, we can take $\ancillaket = \sum_n \sqrt{n} \ket{n} \ket{n}$. This would incur a dimension increase of at most from $d \rightarrow d^2$ and involve rewriting the unitary $U$ as $\id \otimes U$, but in our work we will leave the dimension on the ancilla to be unrestricted, and consider the identity to always be included in the set of free unitaries. We may interchangeably write the channels of the form \cref{eq:general_channel} as $ \text{Tr}_X  (  U ( \rho \otimes \ancillamixed ) U^\dagger  )$ or $ \text{Tr}_X  (  U (\rho \otimes \ancillaproj )U^\dagger  )$.
\end{remark}

Having justified the form for channels considered in our framework, we can now use \cref{eq:general_channel} and \cref{eq:general_subchannel} as a solid foundation as we progress to our results section.

\section{Results}

In this section we now begin presenting in detail our results.
Our first main contribution is to rule out a model of universal quantum computation that involves purely incoherent resources acting on some (possibly coherent) resourceful state. That is, we show that such an example would not exist in \cref{tab:examples}.

We will proceed by considering channels of the form $\mathcal{E}(\rho) = \text{Tr}_X  (  U ( \rho \otimes \ancillamixed ) U^\dagger  )$ as in \cref{eq:general_channel}, and we will compare these to the channel $\rho \mapsto H^{\otimes n} \rho H^{\otimes n}$. As per the discussion above in \cref{sec:framework}, the channel $\rho \mapsto \mathcal{E}(\rho)$ is a mathematical way of writing any operation that involves free unitaries and adaptive computational basis measurements acting on the input state $\rho$ and some fixed ancilla $\ancillamixed$.

We will consider two cases: firstly when we only allow the ability to perform incoherent unitaries (such as CNOT, $S$, $T$ etc.). In this case the unitary $U$ in $\mathcal{E}(\rho)$ will a be quantum controlled version of an incoherent unitary. Secondly, we will tackle the case of using $k$ Hadamards to implement $n>k$ Hadamards. We consider the cases of exact, deterministic, approximate and probabilistic implementation, see \cref{tab:results} for a detailed summary of our findings. 

\newcommand{\bigqm}[1][1]{\text{\larger[#1]{\textbf{?}}}}
\newcommand{\xmark}{\ding{55}}

{\renewcommand{\arraystretch}{2}
\begin{table}[htbp!]
    \centering
    \begin{tabular}{c||c c|c}
         &  \multicolumn{2}{c}{$0 \mapsto n$} &  $k \mapsto n$  \\\hline
        Exact \& Deterministic & \xmark &  (\cref{thm:exact}) & \quad \xmark \quad (\cref{thm:kton}) \\
           Exact \& Probabilistic & \xmark & (\cref{thm:exact}) & \bigqm[2] \\
        Approximate \& Deterministic & $
\mathcal{D}  \geq   1 - \frac{1}{2^n}  $ & (\cref{thm:approx}) &   \bigqm[2] \\
        Approximate \& Probabilistic & $
\mathcal{D} \geq 1 - \frac{1}{2^n} $ & (\cref{thm:approx}) & \bigqm[2]
    \end{tabular} 
    \caption{Summary of our results for using $k$ Hadamards, incoherent unitaries, classical control, computational basis measurements and an arbitrary ancilla to simulate $n$ Hadamards, where $n > k$. A cross (\xmark) indicates that we have proven a no-go result for this case. Here $\mathcal{D} =\max_\rho D( \mathcal{E}(\rho), H^{\otimes n} \rho H^{\otimes n})$ for $D$ trace distance and $\mathcal{E}$ is the simulating channel using $k$ Hadamards, as defined in \cref{ob:map}. A question mark (\bigqm[1]) indicates that we have not considered this case in this work. The approximate bounds should also be compared with the case of using no ancilla, for which we show a lower bound of $\mathcal{D} \geq \sqrt{1-2^{k-n}}$ in \cref{lem:ancilla_indep}.}
    \label{tab:results}
\end{table}}

The following subsection introduces some simple facts and supporting results.

\subsection{Preliminaries} \label{subsec:prelim}

Here we discuss some basic results that will be useful to us in this section. The following lemma will prove crucial.

\begin{lemma} \label{lem:controlled_incoherent}
The family of controlled unitaries $\mathcal{C}(\mathcal{U})$ are incoherent if and only if $\mathcal{U}$ are incoherent.
\end{lemma}

\begin{proof}
Recall from \cref{def:controlledU} that the controlled unitaries are of the form
\[
V = \sum_{x \in S} \ketbra{x}{x} \otimes U + \sum_{y \in S^c }\ketbra{y}{y} \otimes \mathbbm{1}
\]
for $S$ some subset of bitstrings, and $S^c$ its complement.

Now consider this operator acting on a computational basis state $\ket{c_1} \ket{c_2}$, (with the same tensor product structure as $V$ above). First suppose that $c_1 \in  S^c$. Then $V\ket{c_1 c_2} = \ket{c_1 c_2}$. Now consider $c_1 \in  S$. For $U = \sum_z e^{i \theta_z} \ketbra{\pi(z)}{z}$ incoherent, we have 

\begin{align}
V \ket{c_1 c_2}=&\sum_{x \in S} \ketbra{x}{x} \otimes U \ket{c_1 c_2} \\
=&\sum_{x \in S} \ketbra{x}{x} \otimes \bigg ( \sum_z e^{i \theta_z} \ketbra{\pi(z)}{z} \bigg ) \ket{c_1 c_2} \\
=& e^{i \theta_{c_2}} \ket{c_1} \ket{\pi(c_2)} .
\end{align}
Hence $V$ is of the form $V= \sum_x e^{i \theta_x} \ketbra{\pi(x)}{x}$ and is thus incoherent.
For the other direction, if $U$ is not incoherent, then it will map at least one basis vector $\ket{c_2}$ to a superposition. Then for $c_1 \in S$, $V$ will map $\ket{c_1 c_2}$ to $\ket{c_1}U\ket{c_2}$, which will also be a superposition, and so $V$ is not incoherent.
\end{proof}

For our purposes, this Lemma allows us to take $U$ to be itself incoherent in \cref{ob:map}, as by definition it belongs to the set of controlled incoherent unitaries. Note also that as SWAP is itself an incoherent unitary, in our case we can without loss of generality take the trace in \cref{ob:map} to be on the ancillary subsystem, i.e.
\[
    \mathcal{E}^x(\rho) = \text{Tr}_2 \bigg ( U ( \rho \otimes \ancillamixed )  U^\dagger \bigg ).
\]

In addition, we will consider the case of being able to perform $k$ Hadamard gates, and seek to use incoherent resources and a supplementary state to implement $n>k$ Hadamard gates. In this case, without loss of generality the unitary $U$ above will be of the following form
\[
U=U_k V_k \dots U_1 V_1 U_0, \label{eq:U_iV_i}
\] for $U_i$ incoherent unitaries and $V_i$ controlled-Hadamards (``controlled'' in the general sense of \cref{def:controlledU}).

The following definition serves as a useful discrete quantifier of coherence for pure states.

\begin{definition}  \label{def:coherencerank}
    The \textit{coherence rank} \cite{killoran2016converting, streltsov2017colloquium} of a pure state $\ket{\psi}$ is defined to be the minimum number of terms required to write the state as a linear combination of computational basis states. We denote this by $\chi(\ket{\psi})$.
\end{definition}
For example $\chi(\ket{x})=1$ for any computational basis state $\ket{x}$, and $\chi(\ket{+}^{\otimes n})=2^n$. We also have that $\chi(\ket{\psi}\otimes\ket{\phi})= \chi(\ket{\psi})\chi(\ket{\phi})$. With this defined, we can state the following lemma.

\begin{lemma} \label{lem:cohrankprod}
    Let $U=U_k V_k \dots U_1 V_1 U_0$ be a product of unitaries, alternating between incoherent unitaries $U_i$ and controlled-Hadamards $V_i$. Then we have for any state $\ket{\psi}$, the coherence rank satisfies
    \[
    \frac{\chi(\ket{\psi})}{2^k} \leq \chi \left (U \ket{\psi} \right ) \leq 2^k \chi(\ket{\psi}).  
    \]
\end{lemma}

\begin{proof}
First note that $\chi(U\ket{\psi}) = \chi(\ket{\psi})$ for any incoherent unitary $U$ (they can only permute and apply local phases to computational basis states). Now let $V$ be a Hadamard or controlled-Hadamard gate (in the general sense of \cref{def:controlledU}), and consider the action on a computational basis state $\ket{x}$. We have that $ V \ket{x}$ must have coherence rank either $1$ or $2$. Write $\ket{\psi} = \sum_x \alpha_x \ket{x}$, from which it becomes clear that $V \ket{\psi} = \sum_x \alpha_x V \ket{x}$ can have at most $2\chi(\ket{\psi})$ terms (some of the terms could cancel). Hence we have $\chi (V \ket{\psi}) \leq 2 \chi(\ket{\psi})$ for all $\ket{\psi}$, which also implies that $\chi (\ket{\psi}) \leq 2 \chi(V^\dagger \ket{\psi}) = 2 \chi(V \ket{\psi})$, as $V$ is self-inverse. Combining these shows that
 \[
    \frac{\chi(\ket{\psi})}{2} \leq \chi \left (V \ket{\psi} \right ) \leq 2 \chi(\ket{\psi}).
    \]

    Now we can use induction. The base case of $U=U_1 V_1$ follows immediately from the above. Now suppose that
     \[
    \frac{\chi(\ket{\psi})}{2^k} \leq \chi \left (U_k V_k \dots U_1 V_1 U_0 \ket{\psi} \right ) \leq 2^k \chi(\ket{\psi})  .
    \]
    Define $\ket{\phi} = U_k V_k \dots U_1 V_1 U_0 \ket{\psi}$, and we then get
 \[
    \frac{\chi(\ket{\phi})}{2} \leq \chi \left (U_{k+1}V_{k+1} \ket{\phi} \right ) = \chi \left (V_{k+1} \ket{\phi} \right )  \leq 2 \chi(\ket{\phi}).
    \]

\end{proof}
We will use these upper and lower bounds on the coherence rank as a key ingredient in \cref{thm:kton}, one of our no-go results.

We can also make some simple observations about the trace distance, in particular:

\begin{restatable}{lemma}{tracedistancepure}
 \label{lem:tracedistancepure}
    For the induced trace distance, it is sufficient to take the maximum over pure states, i.e. for any channels $\mathcal{E}$, $\mathcal{V}$
    \[
    \mathcal{D}  ( \mathcal{E}, \mathcal{V}  ) := \max_\rho ~ D\bigg ( \mathcal{E}(\rho), \mathcal{V}(\rho)\bigg ) =   \max_{\ket{\phi} } ~ D \bigg ( \mathcal{E}(\ketbra{\phi}), \mathcal{V}(\ketbra{\phi}) \bigg ) .
    \]
\end{restatable}
\begin{proof}
    \begin{align}
         \max_\rho D ( \mathcal{E}(\rho), \mathcal{V}(\rho)) &=  \max_{p_i, \psi_i} D \bigg ( \sum p_i \mathcal{E}( \psi_i), \sum p_i \mathcal{V}( \psi_i) \bigg ) \\
         &\leq \max_{p_i, \psi_i} \sum p_i D(  \mathcal{E}( \psi_i),  \mathcal{V}( \psi_i)) \\
         &\leq \max_{p_i, \psi_i} \sum p_i  \left ( \max_{\phi} D( \mathcal{E}(\phi), \mathcal{V}(\phi)) \right ) \\
         &=  \max_{\phi} D( \mathcal{E}(\phi), \mathcal{V}(\phi)) \\
         & \leq  \max_\rho D( \mathcal{E}(\rho), \mathcal{V}(\rho))
    \end{align}
where $\rho$ denotes a mixed state,  $\psi_i$ and $\phi$ denote pure states, and we used linearity of the channels and joint convexity of the trace distance \cite{nielsen2002quantum}.
\end{proof}

For pure states we also have that  \cite{wilde2013quantum}
\[
D \bigg (\ketbra{\psi}, \ketbra{\phi} \bigg ) = \sqrt{1 - \abs{\braket{\psi}{\phi}}^2}.
\]
Combining this with \cref{lem:tracedistancepure} leads to the following corollary:

\begin{corollary} \label{cor:tdistpure}
    The induced trace distance between unitary channels is given by
    \begin{align}
     \mathcal{D}( U, V ) &= \max_{\ket{\psi}} \sqrt{1- \abs{\bra{\psi}U^\dagger V\ket{\psi}}^2}.
    \end{align}
\end{corollary}

Let us now briefly consider the case of no-ancilla. We would expect a non-zero distance between $H^{\otimes n}$ and a unitary composed of incoherent gates and $k <n$ Hadamard gates. We have the following lower bound in this case:

~ 

\begin{lemma}

Let $U=U_k V_k \dots U_1 V_1 U_0 $, where $U_i$ are incoherent and $V_i$ are controlled \\ Hadamards, and let $n \geq k$. Then
    \begin{align}
    \mathcal{D} ( U, H^{\otimes n}) \geq \sqrt{1 - 2^{k-n}}.
\end{align}
\end{lemma}

\begin{proof}
Using \cref{cor:tdistpure} we have that
\begin{align}
\mathcal{D}( U, H^{\otimes n} ) &= \max_{\ket{\psi}} \sqrt{1- \abs{\bra{\psi}H^{\otimes n} U\ket{\psi}}^2} \\
&\geq \sqrt{1- \abs{\bra{0^n}H^{\otimes n} U\ket{0^n}}^2} \\
&\geq \sqrt{1- \abs{\bra{+^n}U\ket{0^n}}^2}.
\end{align}
Now $U\ket{0^n} = U_k V_k \dots U_1 V_1 U_0 \ket{0^n}$ will have coherence rank at most $2^k$ (as each incoherent unitary preserves the coherence rank, and each controlled Hadamard can at most double the coherence rank, by  \cref{lem:cohrankprod}). Thus we have $U\ket{0^n} = \sum_{x=0}^{2^k-1} \alpha_x \ket{c_x}$ where $\ket{c_x}$ are (not necessarily distinct) computational basis states. As $\ket{+^n} = 2^{\frac{-n}{2}} \sum_{x \in \{ 0,1 \}^n} \ket{x}$, we have that
\[
\abs{\bra{+^n}U\ket{0^n}}^2 = \abs{\frac{\sum_{x=0}^{2^k-1} \alpha_x}{2^{\frac{n}{2}}}}^2
\leq\frac{2^k \sum_{x=0}^{2^k-1}  \abs{\alpha_x}^2}{2^n}
 = 2^{k-n},
\]
where the inequality follows from Cauchy-Schwarz: $\abs{\sum_{x=0}^{M-1} \alpha_x}^2 \leq M \sum_{x=0}^{M-1} \abs{\alpha_x}^2$ for all $M$, and we used $\sum_{x=0}^{2^k-1} \abs{\alpha_x}^2 = 1$. Hence we have that
\begin{align}
\mathcal{D}( U, H^{\otimes n} ) &\geq \sqrt{1- \abs{\bra{+^n}U\ket{0^n}}^2} \\
&\geq \sqrt{1- 2^{k-n}}.
\end{align}    
\end{proof}

This shows that if $k \ll n$, the unitary $U$ composed of $k$ Hadamards cannot be close in induced trace distance to $n$ Hadamards, as intuitively expected.

We now consider the first case of using purely incoherent unitaries ($0$ Hadamards) and an ancilla to implement $n$ Hadamards, considering the exact, approximate and probabilistic cases. 

\subsection{Incoherent resources and an ancilla cannot implement $n>0$ Hadamards}

We first consider the question of whether incoherent resources supplemented with an arbitrary ancillary state can implement a single Hadamard gate. We show that this is not the case. Our strategy is to first show  in \cref{lem:exact} below that if a channel satisfies a certain relation with the dephasing map, then when acting jointly on an input state and fixed arbitrary ancilla the channel cannot implement any coherent unitary. Secondly, we show in \cref{lem:commutes} that channels that only use incoherent resources supplemented with an arbitrary ancilla state (see \cref{sec:framework}) satisfy this relation, and hence are not able to implement any coherent unitary, such as the Hadamard.

To begin, let us prove the following lemma, which is in fact a generalisation of an observation made in the erratum of \cite{dana2017resource} (see \cref{ex:erratum} for more context here).

\begin{lemma}\label{lem:exact}
Let $\mathcal{E} : \mathcal{S}(\mathcal{H}_1) \otimes  \mathcal{S}(\mathcal{H}_2) \rightarrow  \mathcal{S}(\mathcal{H}_1)$ be any channel such that
\[
 \Delta \circ \mathcal{E} \circ \Delta = \Delta \circ \mathcal{E},\label{eq:thm1condition}
\]
where $\Delta$ is the dephasing map defined in \cref{eq:dephasing}.
 Then for any state $\ancillamixed \in  \mathcal{S}(\mathcal{H}_2)$ the channel $\mathcal{E}_\ancillamixed (\rho):= \mathcal{E}(\rho \otimes \ancillamixed)$ cannot implement any coherent unitary exactly.
\end{lemma}

\begin{proof}[Proof of \cref{lem:exact}]
Suppose that $\mathcal{E}(\rho \otimes \ancillamixed) = U \rho U^\dagger$ for some unitary $U$. We seek to show that $U$ must be incoherent if the condition (\ref{eq:thm1condition}) is met. This condition implies that

\[\label{eq:rhocondition}
\Delta  \bigg (\mathcal{E}(\Delta(\rho \otimes \ancillamixed))  \bigg ) = \Delta  \bigg (U \rho U^\dagger \bigg ), \qquad \forall \rho.
\]

Let $U_{yx}:= \bra{y}U\ket{x}$, so that
\begin{align}
U\ket{x} &= \sum_y U_{yx} \ket{y} &&U^\dagger \ket{x} = \sum_y U_{xy}^* \ket{y} \\
\bra{x}U &= \sum_y U_{xy} \bra{y} &&\bra{x}U^\dagger = \sum_y U_{yx}^* \bra{y},
\end{align}
We can now use the facts that $\Delta(\rho \otimes \ancillamixed) = \Delta(\rho) \otimes \Delta(\ancillamixed)$ and $\Delta (\ketbra{x}{x}) = \ketbra{x}{x}$, and first input $\rho = \ketbra{x}$ in \cref{eq:rhocondition} to obtain
\begin{align}
      \Delta \bigg ( \mathcal{E}  ( \ketbra{x} \otimes \Delta(\ancillamixed)) \bigg ) &= \Delta \left( \sum_{y,z} U_{yx}U_{zx}^* \ketbra{y}{z}  \right )\\
      &= \sum_y \abs{U_{yx}}^2 \ketbra{y}
      \end{align}
Multiplying both sides by $\abs{U_{zx}}^2$ and summing over $x$ implies
      \begin{align}
           \sum_x \abs{U_{zx}}^2 \Delta \bigg ( \mathcal{E}   ( \ketbra{x} \otimes \Delta(\ancillamixed)) \bigg) &= \sum_{x,y} \abs{U_{zx}}^2  \abs{U_{yx}}^2 \ketbra{y} .\label{eq:firsteq}
\end{align}
Now focusing on $ \rho = U^\dagger\ketbra{z}U = \sum_{x,w} U_{zw} U_{zx}^*\ketbra{x}{w} $, we have $\Delta (\rho) = \sum_{x}\abs{U_{zx}}^2 \ketbra{x}{x} $. \cref{eq:rhocondition} then implies that
\begin{align}
    \sum_x \abs{U_{zx}}^2 \Delta \circ \mathcal{E} \bigg(  \ketbra{x} \otimes \Delta(\ancillamixed)\bigg ) &= \ketbra{z} . \label{eq:secondeq}
\end{align}
Comparing (\ref{eq:firsteq}) with (\ref{eq:secondeq}), as the left-hand sides are equal we see that 
\begin{align}
      \sum_{x,y} \abs{U_{zx}}^2  \abs{U_{yx}}^2 \ketbra{y} &= \ketbra{z}. 
\end{align}
      This can only be true if
\begin{align}
      \sum_{x} \abs{U_{zx}}^2  \abs{U_{yx}}^2  = 0 \quad \text{for $y\neq z$}.
      \end{align}
However as all terms are non-negative, we have the stronger implication that
      \begin{align}
          \abs{U_{zx}}  \abs{U_{yx}}  &= 0 \quad \text{for $y\neq z$, $\forall x$.}
\end{align}
This final equation implies that for all $x$, there can be at most one value of $y$ such that $U_{yx}\neq 0$. As $U$ is unitary, this implies that in each column there is exactly one non-zero entry, so $U$ must be incoherent. In particular, 
\[
U\ket{x} = \sum_y U_{yx} \ket{y} = U_{y'x} \ket{y'},
\]
for some $y'$ (depending on $x$). Thus in summary,  \cref{eq:rhocondition} implies that if the channel $\mathcal{E}_\ancillamixed$ is unitary then it must be incoherent, independent of $\ancillamixed$.
\end{proof}

To put this technical result into context, we now show that the channels proposed in \cref{ob:map} satisfy the condition in \cref{lem:exact}, and thus an arbitrary ancilla is not sufficient to elevate incoherent resources to computational universality.

\begin{lemma} \label{lem:commutes}
For any incoherent unitary $U$, the map $\rho \mapsto\text{Tr}_2 \left (U \rho U^\dagger \right )$ commutes with the dephasing map $\Delta$.
\end{lemma}
See \cref{eq:incoherent_unitary} for the definition of incoherent unitaries, and \cref{eq:dephasing} for the definition of the dephasing map.
\begin{proof}
First let us see that the dephasing map commutes with the action of any incoherent unitary, recall that these are of the form $U = \sum_{x=1}^d e^{i \theta_x} \ketbra{\pi(x)}{x} $ for some permutation $\pi$ (\cref{eq:incoherent_unitary}).

Then for any state $\rho$ we have

\begin{align}
     U \Delta ( \rho ) U^\dagger &= \left (   \sum_{y} e^{i \theta_y} \ketbra{\pi(y)}{y} \right )  \sum_x \ketbra{x} \rho   \ketbra{x} \left ( \sum_{z} e^{-i \theta_z} \ketbra{z}{\pi(z)} \right )  \\
    &=  \sum_{x, y, z} e^{i \theta_y} e^{-i \theta_z} ~  \ket{\pi(y)} \! \braket{y}{x} \! \bra{x}   \rho   \ket{x} \! \braket{x}{z} \! \bra{\pi(z)}  \\
    &=  \sum_{x} ~  \ketbra{\pi(x)}{x}    \rho     \ketbra{x}{\pi(x)}.  \label{eq:pi1}
\end{align}
Noting that we can also write the dephasing map as $ \Delta (\rho) = \sum_x \ketbra{\pi(x)} \rho \ketbra{\pi(x)} $ for any permutation $\pi$ leads to
\begin{align}
    \Delta (U \rho U^\dagger) &= \sum_x \ketbra{\pi(x)} \left (   \sum_{y} e^{i \theta_y} \ketbra{\pi(y)}{y} \rho  \sum_{z} e^{-i \theta_z} \ketbra{z}{\pi(z)} \right )\ketbra{\pi(x)} \\
    &=  \sum_{x, y, z} e^{i \theta_y} e^{-i \theta_z}  \ket{\pi(x)}\hspace{-2pt}\braket{\pi(x)}{\pi(y)}\hspace{-2pt}\bra{y} \rho   \ket{z}\hspace{-2pt}\braket{\pi(z)}{\pi(x)}  \hspace{-2pt}\bra{\pi(x)}\\
    &=  \sum_{x} ~  \ketbra{\pi(x)}{x}    \rho     \ketbra{x}{\pi(x)}. \label{eq:pi2}
\end{align}
Hence the equality of \cref{eq:pi1} and \cref{eq:pi2} (and as $\rho$ was arbitrary) show that $\Delta \circ U = U \circ \Delta$ for any incoherent $U$. It is also clear that the dephasing map commutes with the partial trace, from which the result follows.
\end{proof}

Since $\Delta^2 = \Delta$, we have that $  \mathcal{E} \circ \Delta = \Delta \circ \mathcal{E} \implies  \Delta \circ \mathcal{E} \circ \Delta = \Delta \circ \mathcal{E}$ for any channel $\mathcal{E}$, that is, commutation of $\mathcal{E}$ with $\Delta$ implies the condition imposed in \cref{lem:exact}.

Hence \cref{lem:commutes} and \cref{lem:exact} together show that given incoherent unitaries and classical control, encapsulated by the channel $\rho \mapsto \mathcal{E}(\rho \otimes \ancillamixed)$ (see the discussion in \cref{sec:framework}), cannot exactly implement any coherent unitary, even when supplemented with an arbitrary ancilla. This is in direct contrast to other situations, such as magic state injection.

This result is about exactly implementing a coherent unitary. It is natural to question whether this no-go result arises from demanding too much. One possible relaxation is to consider implementation of a coherent unitary with some non-zero probability: surprisingly we can still show that this is impossible. 

We can also show the same result for probablistic implementations. Recall from \cref{ob:map} (and surrounding text) that we represent these by convex combinations of normalised sub-channels

\[
    \mathcal{E}^x(\rho) = \alpha \text{Tr}_X \bigg ( \id \otimes \ketbra{x} U \rho \otimes \ancillamixed  U^\dagger \bigg ).
\]

where the normalisation factor $\alpha$ does not depend on the input $\rho$.

\begin{restatable}{lemma}{subchannelcommuteapprox}
\label{lem:subchannel_commutes_approx}
For incoherent $U$, the normalised sub-channel $\rho \mapsto \mathcal{E}^x (\rho)$ commutes with the dephasing map $\Delta$.
\end{restatable} 
The working is very similar to that of \cref{lem:commutes}, and we give a more concise proof as follows.
\begin{proof}

\begin{align}
    (\Delta \circ \mathcal{E}^x) (\rho) &= \sum_i \ketbra{i} \text{Tr}_X \bigg ( \id \otimes \ketbra{x} U \rho U^\dagger \bigg ) \ketbra{i} \\
     &= \sum_j \text{Tr}_X \bigg ( \ketbra{j} \id \otimes \ketbra{x} U \rho U^\dagger \ketbra{j} \bigg ) \\
      &= \text{Tr}_X \bigg (  \id \otimes \ketbra{x} \sum_j \ketbra{j} U \rho U^\dagger \ketbra{j} \bigg )    \\
      &= \text{Tr}_X \bigg (  \id \otimes \ketbra{x} U \sum_j \ketbra{j}  \rho  \ketbra{j} U^\dagger \bigg )   \\
      &= (\mathcal{E}^x \circ \Delta )(\rho) 
\end{align}
\end{proof}

As any convex combination of such channels will also commute with the dephasing map, by \cref{lem:exact} we can see that any attempt to even probabilistically implement a coherent unitary exactly will fail. We can summarise the preceding with our first main result:

\begin{theorem} \label{thm:exact}
Given the ability to perform incoherent unitaries, computational basis measurements and classical control, it is impossible to implement any coherent unitary (e.g. Hadamard) exactly with any non-zero probability, even when supplemented with an arbitrary ancilla.
\end{theorem}
\begin{proof}
    As discussed in \cref{sec:framework}, the ability to perform incoherent unitaries, computational basis measurements and classical control is encapsulated by channels of the form $\rho \mapsto \text{Tr}_2(U\rho U^\dagger)$ for $U$ incoherent, or in the probabilistic case by convex combinations of subchannels $\rho \mapsto \text{Tr}_2(\id \otimes \ketbra{x} ~ U\rho U^\dagger)$. By \cref{lem:commutes} and \cref{lem:subchannel_commutes_approx}, these maps commute with the dephasing map $\Delta$. Hence we can apply \cref{lem:exact} to see that given access to the above operations, one can never implement any coherent unitary exactly.
\end{proof}

\subsubsection*{Approximate Implementation}

Having considered the exact and probabilistic cases, we now turn our attention to the approximate case, focusing on tensor products of Hadamard gates (as opposed to arbitrary coherent unitaries). Specifically, we seek lower bounds on the induced trace distance between the channels introduced in \cref{ob:map} and $n$ Hadamard gates. Our second main technical result achieves this goal as follows.

\begin{lemma}
    \label{lem:0ton_approx}
    Let $\mathcal{E} : \mathcal{S}(\mathcal{H}_1) \otimes  \mathcal{S}(\mathcal{H}_2) \rightarrow  \mathcal{S}(\mathcal{H}_1)$ be any channel that commutes with the dephasing map, i.e. 
\[
 \mathcal{E} \circ \Delta = \Delta \circ \mathcal{E}, 
\]
where $\Delta$ is the dephasing map defined in \cref{eq:dephasing}.
Define the channel $\mathcal{E}_\ancillamixed (\rho) := \mathcal{E}(\rho \otimes \ancillamixed)$ for an arbitrary state $\ancillamixed \in \mathcal{S}(\mathcal{H}_2) $. Let $\mathcal{D}$ denote the induced trace distance on quantum channels. Then for all states $\ancillamixed$, we have
\[
\mathcal{D}\bigg (\mathcal{E}_\ancillamixed ~,~ H^{\otimes n} \bigg) \geq 1 - \frac{1}{2^n}.
\]
\end{lemma}

\begin{proof}[Proof of \cref{lem:0ton_approx}.]

Let $C_n = \{\ket{x} ~:~ x \in \{0,1\}^n \}$ denote the set of computational basis states, and $B_n = \{\ket{x} ~:~ x \in \{+,-\}^n \}$ denote the set of conjugate basis states. Note that $H^{\otimes n}$ maps bijectively between $C_n$ and $B_n$. We then have

\begin{align}
\mathcal{D}\bigg (\mathcal{E}_\ancillamixed~,~ H^{\otimes n} \bigg) &:= \max_\rho D\bigg ( \mathcal{E} (\rho \otimes \ancillamixed)~,~ H^{\otimes n}\rho H^{\otimes n} \bigg ) \\
    &\geq  \max_\rho D\bigg ( \Delta \circ \mathcal{E} (\rho \otimes \ancillamixed)~,~ \Delta (H^{\otimes n}\rho H^{\otimes n}) \bigg) \label{eq:distance_contractivity}\\
    &=  \max_\rho D \bigg(\mathcal{E} (\Delta(\rho) \otimes \Delta(\ancillamixed))~,~ \Delta (H^{\otimes n}\rho H^{\otimes n})\bigg) \\
    &\geq  \max_{\ket{\phi} \in B_n} D \bigg(\mathcal{E} (\Delta(\ketbra{\phi}) \otimes \Delta(\ancillamixed))~,~ \Delta (H^{\otimes n}\ketbra{\phi} H^{\otimes n})\bigg) \\
    &=  \max_{\ket{\psi} \in C_n} D \bigg(\mathcal{E} (\frac{\id}{2^n} \otimes \Delta(\ancillamixed))~,~ \ketbra{\psi} )\bigg),
\end{align}
where we used the contractivity of the trace distance under quantum channels, and the condition on $\mathcal{E}$ commuting with the dephasing map from the theorem statement.

Now define $\sigma := \mathcal{E} (\frac{\id}{2^n} \otimes \Delta(\ancillamixed))$, and note that $\Delta (\ketbra{\psi}) = \ketbra{\psi}$ for all $\ket{\psi} \in C_n$. Again using contractivity of the trace distance we can then write

\begin{align}
     \mathcal{D}\bigg (\mathcal{E}_\ancillamixed~,~ H^{\otimes n} \bigg) & \geq \max_{\ket{\psi} \in C_n} D \bigg( \sigma ~,~ \ketbra{\psi} )\bigg) \\
      &\geq  \max_{\ket{\psi} \in C_n} D 
      \bigg( \Delta ( \sigma ) ~,~ \Delta ( \ketbra{\psi}) )\bigg ) \\
      &=  \max_{\ket{\psi} \in C_n} D 
      \bigg( \Delta ( \sigma ) ~,~ \ketbra{\psi} )\bigg ) \\
      &\geq \max_{\ket{\psi} \in C_n} \bigg (1 - \bra{\psi} \Delta (\sigma) \ket{\psi} \bigg ) \\
       &\geq 1 - \frac{1}{2^n}
\end{align}
The last line can be seen by observing that for any incoherent state, the maximum diagonal entry must be at least $\frac{1}{2^n}$.

\end{proof}

This bound is displayed in \cref{tab:results}, and for the case of a single Hadamard the bound becomes $\mathcal{D} ( \mathcal{E}_\ancillamixed, H ) \geq \frac{1}{2}$. Operationally, this means that using an optimal (unentangled) input to the channels, we can distinguish them with high probability given multiple uses. Note also that as the induced trace distance is a lower bound on the diamond distance \cite{watrous2018theory}, we also have the same lower bound on the diamond distance between the above channels. We can also see that this bound is tight, as for example taking the channel $\mathcal{E}$ to be the map that always outputs the maximally mixed state (which commutes with $\Delta$), we have that $\mathcal{D}(\mathcal{E}_\tau, H^{\otimes n} ) = 1 - \frac{1}{2^n}$ which matches the bound.

Furthermore, we can observe that the above analysis also applies to the probabilistic case: using the fact that the corresponding normalised subchannel  (\ref{eq:general_subchannel}) commutes with $\Delta$ (\cref{lem:subchannel_commutes_approx}) we see that \cref{lem:0ton_approx} also applies. The conclusion is that even approximate, probabilistic implementation of Hadamards is not possible, which is our second main result. Recall (from \cref{subsec:defns}) that we say we can implement a channel $\mathcal{E}$ \textit{$\epsilon$-approximately} if we can implement a channel $\mathcal{V}$ with induced trace distance $\mathcal{D}(\mathcal{E}, \mathcal{V}) \leq \epsilon$.

\begin{theorem} \label{thm:approx}
Given the ability to perform incoherent unitaries, computational basis measurements and classical control, it is impossible to implement $n$ Hadamards $\epsilon$-approximately with non-zero probability, for $0 \leq \epsilon <  1 - 2^{-n} $. In particular, it is impossible to implement a single Hadamard $\epsilon$-approximately with non-zero probability, for $0 \leq \epsilon < \frac{1}{2}$.
\end{theorem}
\begin{proof}
The proof follows a similar structure to that of \cref{thm:exact}.
   We can describe channels arising from the stated operations by channels of the form $\rho \mapsto \text{Tr}_2(U\rho U^\dagger)$ for $U$ incoherent, or in the probabilistic case by convex combinations of subchannels $\rho \mapsto \text{Tr}_2(\id \otimes \ketbra{x} ~ U\rho U^\dagger)$ (see \cref{sec:framework}). These maps commute with the dephasing map $\Delta$ by \cref{lem:commutes} and \cref{lem:subchannel_commutes_approx}. Then \cref{lem:0ton_approx} implies that given access to the above operations, one can never implement a channel that has induced trace distance with $H^{\otimes n}$ of strictly less than $ 1 - 2^{-n} $.
\end{proof}

\subsection{Incoherent resources, $k$ Hadamards and an ancilla cannot implement $n>k$ Hadamards}

The above showed that incoherent resources supplemented with an arbitrary state is not sufficient to implement even a single Hadamard gate, even approximately and probabilistically. A further generalisation is to consider the case of having the ability to perform up to $k$ Hadamard gates, incoherent unitaries, classical control and access to an ancillary state. Could it be possible here to implement $n$ Hadamard gates, with $n$ strictly greater than $k$? It would be very striking if this were the case, as then by repeating this process (and having many copies of the ancilla) one could implement an arbitrarily high number of Hadamard gates when originally only given the ability to perform a fixed number of them. We will show that this is indeed not the case, which demonstrates the robustness of our previous results in this direction. Our result holds when also considering general controlled Hadamard gates, which is stronger than considering single qubit Hadamard gates as these are a special case of our generalised controlled operations in \cref{def:controlledU}. Note also that the previous section is a special case of the scenario here (with $k=0$), however the approach and proof technique here differs substantially. 

Let us first see a simple example of the case $1 \mapsto 2$ to illustrate the general argument to follow. For example, one could ask about the existence of circuits of the form as in \cref{fig:1->2 hadamards}, for all two-qubit inputs $\ket{\psi}$, fixed ancilla $\ancillaket$, incoherent unitaries $U$, $V$ and $W$, and where $H_i$ denotes a Hadamard on the $i$-th qubit.

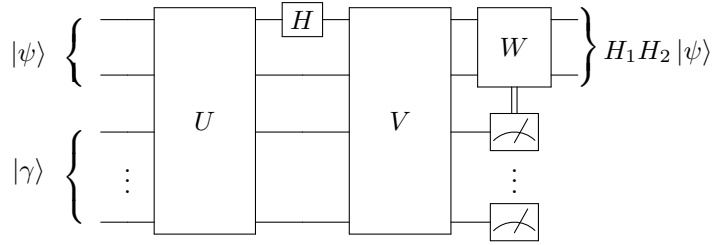
\begin{figure}[htbp!]
    \centering
    $$
  \Qcircuit @C=1em @R=1em {
 &&\qw& \multigate{4}{\quad U \quad} & \gate{H} & \multigate{4}{\quad V \quad} &\multigate{1}{ ~ W ~ } & \qw &&  \\
 &&\qw & \ghost{\quad U \quad} & \qw &  \ghost{\quad V \quad} & \ghost{~W~} & \qw\inputgroupv{1}{2}{0em}{1.3em}{\ket{\psi}}\gategroup{1}{8}{2}{8}{.8em}{\}} &&&  \raisebox{2em}{\ensuremath{H_1 H_2 \ket{\psi}}} \\ 
 &&\qw & \ghost{\quad U \quad} & \qw &  \ghost{\quad V \quad} &\meter \cwx[-1] \\ 
 && \qvdots & && &  \qvdots \\ 
 &&\qw & \ghost{\quad U \quad} & \qw &  \ghost{\quad V \quad} & \meter  \inputgroupv{3}{5}{0em}{1.5em}{\ancillaket} \\
}
$$
\caption{A possible circuit for using a single Hadamard gate and incoherent unitaries $U$, $V$ and $W$ to implement two Hadamard gates. We rule out the existence of such a circuit in this work. }
\label{fig:1->2 hadamards}
\end{figure}

Let us now be more precise, and argue by contradiction. Suppose that the above task was possible. This would mean that the following equation would hold:
\[
 H_1 H_2 \ket{\psi}\ket{\ancillapure_\psi} = V H_1 U \ket{\psi}\ancillaket,
\]
 for some state $\ancillaket$, a set of states $\ket{\ancillapure_\psi}$ that could depend on $\ket{\psi}$, incoherent unitaries $U$ and $V$, and where  $H_i$ denotes a Hadamard on the $i$-th qubit. To see how the above diagram can be written in this form, recall from  \cref{sec:framework} that we can replace the last gate $W$ by its quantum controlled version (which will be incoherent by \cref{lem:controlled_incoherent}) and absorb it into $V$.

Now expanding $\ket{\psi}$ in a basis allows us to see that $\ket{\ancillapure_\psi}$ must actually be independent of $\ket{\psi}$. One may already expect this, as the first register contains all the information of the pure state $\ket{\psi}$, for the ancilla system to contain some information of $\ket{\psi}$ would seem to imply a form of cloning. Indeed for $\ket{\psi} = \sum_x \alpha_x \ket{x}$ we have
\[
H_1 H_2 \ket{\psi}\ket{\ancillapure_\psi} = \sum_x \alpha_x V H_1 U \ket{x}\ancillaket = \sum_x \alpha_x H_1 H_2 \ket{x}\ket{\ancillapure_x}.
\]
Tracing out the first register then gives
\[
\ketbra{\ancillapure_\psi} = \sum_x \abs{\alpha_x}^2 \ketbra{\ancillapure_x},
\]
from which the independence on $\ket{\psi}$ follows: a sum of rank 1 operators with non-negative coefficients can only be equal to a rank 1 operator if all the operators are proportional. We can now write 
\[
V H_1 U \ket{\psi}\ancillaket = H_1 H_2 \ket{\psi}\ket{\ancillapure'} \label{eq:example_ancilla'}
\]
for some state $\ket{\ancillapure'}$. Let the coherence rank of $\ancillaket$ and $\ket{\ancillapure'}$ be $r$ and $r'$ respectively. Then taking $\ket{\psi}$ to be $\ket{00}$ and $\ket{++}$, \cref{eq:example_ancilla'} implies the following two equations
\begin{align}
    V H_1 U \ket{00}\ancillaket &= \ket{++}\ket{\ancillapure'} \\
    V H_1 U \ket{++}\ancillaket &= \ket{00}\ket{\ancillapure'}.
\end{align}
By recalling that incoherent unitaries cannot change the coherence rank, and a single Hadamard can at most double and at least halve the coherence rank, these equations respectively imply that

\begin{align}
    4r' &\in  [\frac{r}{2}, 2r] \\
    r' &\in [2r, 8r] \quad \implies 4r' \in [8r, 32r],
\end{align}
which is a contradiction. We thus see that it must be impossible for a single Hadamard gate, incoherent unitaries and classical control to implement two Hadamard gates, even given access to the arbitrary state $\ancillaket$.

We can generalise and formalise this argument to show the impossibility of $k$ Hadamards and incoherent resources implementing $n>k$ Hadamards. The argument will proceed in the same two steps as above: first showing that the ancillary register must be left in a state that is independent of the input $\rho$, and secondly use the resource content of these states (coherence rank) to derive a contradiction.

\begin{lemma} \label{lem:ancilla_indep}
Suppose that there exists some bipartite unitary $U$, local unitary $V$, and state $\ancillaket$ such that for all $\rho$ we have
\[
\text{Tr}_2\left ( U \rho \otimes \ancillaproj U^\dagger \right ) = V \rho V^\dagger.\label{eq:traceequalsV}
\]
Then 
\[
\text{Tr}_1\left ( U \rho \otimes \ancillaproj U^\dagger \right ) = \ketbra{\ancillapure '}
\]
for some fixed pure state $\ket{\ancillapure '}$ that is independent of $\rho$.
\end{lemma}
 
\begin{proof}
    In particular, for $\rho = \ketbra{\psi}$ a pure state, \cref{eq:traceequalsV} becomes
    \[
\text{Tr}_2 \left ( U \ketbra{\psi} \otimes \ancillaproj U^\dagger \right ) = V \ketbra{\psi} V^\dagger.
\]
Since tracing out the second system of $U\ket{\psi}\ancillaket$ results in a pure state for all $\ket{\psi}$, the total state $U\ket{\psi}\ancillaket$ must be a pure product state. So in summary we must have
\[
 U \bigg (\ket{\psi} \ancillaket \bigg ) = V \otimes \id \bigg ( \ket{\psi} \ket{\ancillapure_\psi} \bigg  )
\]
for some pure state $\ket{\ancillapure_\psi}$ that a priori could depend on $\ket{\psi}$.

Now as before, write $\ket{\psi}=\sum_x \alpha_x \ket{x}$. Then we have
\begin{align}
      V \otimes \id \ket{\psi} \ket{\ancillapure_\psi}  = \sum_x \alpha_x ~ U \ket{x}\ancillaket 
= \sum_x \alpha_x ~ V \otimes \id  \ket{x}\ket{\ancillapure_x}
\end{align}
Multiplying by $V^\dagger \otimes \id$ now implies that
\[
 \sum_x \alpha_x \ket{x} \ket{\ancillapure_\psi}  = \sum_x \alpha_x ~ \ket{x}\ket{\ancillapure_x}.
\]

Tracing out the first system gives
\begin{align}
\ketbra{\ancillapure_\psi} = \sum_x \abs{\alpha_x}^2 \ketbra{\ancillapure_x},
\end{align}
which in turn implies
\[
\ketbra{\ancillapure_\psi} =  \ketbra{\ancillapure_x} =  \ketbra{\ancillapure_{x'}} \quad \quad \forall x, x', \psi 
\]
and so $\ket{\ancillapure_\psi}$ must be independent of $\ket{\psi}$.

Hence we have shown that for all pure states $\ket{\psi}$
    \[
\text{Tr}_1\left ( U \ketbra{\psi} \otimes \ancillaproj U^\dagger \right ) = \ketbra{\ancillapure '}
\]
for some state $\ket{\ancillapure '}$ independent of $\ket{\psi}$.
To complete the proof, write an arbitrary mixed state as $\rho = \sum_k p_k \ketbra{\psi_k}$ with $\sum_k p_k = 1$, and observe that
\begin{align}
\text{Tr}_1\left ( U \rho \otimes \ancillaproj U^\dagger \right ) &= \sum_k p_k \text{Tr}_1\left ( U \ketbra{\psi_k} \otimes \ancillaproj U^\dagger \right ) \\
&= \sum_k p_k \ketbra{\ancillapure '} \\
&= \ketbra{\ancillapure '}
\end{align}
\end{proof}

We can use this lemma to explicitly rule out the possibility of $k$ Hadamards and incoherent resources exactly implementing $n$ Hadamards. To do this, we will consider unitaries of the following form:
\[
U=U_k V_k \dots U_1 V_1 U_0,
\]
where $U_i$ are incoherent unitaries and $V_i$ are Hadamards or controlled Hadamards (in the general sense of \cref{def:controlledU}). As discussed in \cref{sec:framework}, this describes any operation involving incoherent unitaries, classical control and $k$ Hadamard gates -- see \cref{eq:U_iV_i} and surrounding text. We can now state our next result.

\begin{lemma} \label{lem:coherencerank}
Let $U=U_k V_k \dots U_1 V_1 U_0$ be a product of unitaries, alternating between incoherent unitaries $U_i$ and controlled-Hadamards $V_i$. If we have that
\[
\text{Tr}_2\left ( U \rho \otimes \ancillaproj U^\dagger \right ) = H^{\otimes n} \rho H^{\otimes n} \qquad \forall \rho, \label{eq:kton_exact}
\]
then we must have that $n \leq k$.
\end{lemma}

\begin{proof}

Consider the case where $\rho = \ketbra{\psi}$ is a pure state. Using \cref{lem:ancilla_indep}, we can then write \cref{eq:kton_exact} as
\[
U  \ket{\psi} \ancillaket   =  H^{\otimes n} \otimes \id \ket{\psi} \otimes \ket{\ancillapure '} \label{eq:pure}
\]
for some fixed state $\ket{\ancillapure '}$.

Let $\ancillaket$ and $\ket{\ancillapure '}$ have coherence ranks $r$ and $r'$ respectively (recall \cref{def:coherencerank} for the definition of the coherence rank $\chi$), and consider the cases of $\ket{\psi}$ being equal to $\ket{0^n}$ and $\ket{+^n}$: 
\begin{align}
    U\ket{0^n}\ancillaket &= \ket{+^n}\ket{\ancillapure '}\\
    U\ket{+^n}\ancillaket  &= \ket{0^n}\ket{\ancillapure '}.
\end{align}
By comparing the coherence ranks of both sides of the above equations, we find that $ \chi(U\ket{0^n}\ancillaket) = 2^n r'$ and $ \chi(U\ket{+^n}\ancillaket) = r'$. Hence
\[
\frac{\chi(U\ket{0^n}\ancillaket)}{2^n} =  \chi(U\ket{+^n}\ancillaket). \label{eq:coherenceranks0+}
\]
We now directly apply \cref{lem:cohrankprod}, which provides lower and upper bounds on the coherence rank of a state after applying a sequence of incoherent unitaries and controlled Hadamards. We obtain
\begin{align}
    \chi (U\ket{0^n}\ancillaket) &\leq 2^k \chi(\ket{0^n}\ancillaket) = 2^k r \\
    \chi (U\ket{+^n}\ancillaket) &\geq 2^{-k} \chi(\ket{+^n}\ancillaket)  = 2^{n-k}r.
\end{align}
Combining these with \cref{eq:coherenceranks0+} leads to
\[
2^{n-k}r  \leq  2^{k-n}r,
\]
which as the coherence rank $r>0$ implies that
\[
n \leq k
\]
as claimed.
\end{proof}

We summarise the implication of the preceding technical result to place it in context with the rest of the paper.

\begin{theorem} \label{thm:kton}
Given the ability to perform incoherent unitaries and $k$ Hadamards, computational basis measurements and classical control, then even with access to an abritray ancillary state, it is impossible to implement $n$ Hadamards exactly, for $n>k$.
\end{theorem}
 \begin{proof}
         As discussed above and in \cref{sec:framework}, any operation involving incoherent incoherent unitaries, $k$ Hadamards, computational basis measurements, classical control and access to some ancillary state $\ancillaket$ can be written as $\rho \mapsto \text{Tr}_2(U \rho \otimes \ancillaproj U^\dagger)$, where  $U=U_k V_k \dots U_1 V_1 U_0$ alternates between incoherent unitaries $U_i$ and controlled Hadamards $V_i$ (controlled in the general sense of \cref{def:controlledU}). \cref{lem:coherencerank} then directly yields the result.
 \end{proof}

\section{Discussion and Open Questions}

We have introduced a unifying framework for approaches to quantum computation involving operations on some fixed, resourceful state. After studying the role of coherence in this context, we showed that some coherence must be present in the operations. By motivating a general form of the possible channels, we have been able to provide a series of no-go results for incoherent resources being able to implement a unitary channel with increased cohering power, even given an arbitrary ancillary state. This shows that unlike e.g. magic, this resource cannot be placed inside a resourceful state and retrieved; it really has to be in the operations, showing a marked difference to other resources for quantum computation. We now detail some avenues for future work.

\subsection{Extending Our Results}
Firstly, it would be of value to extend and sharpen our specific results. For example, we did not include the case of using $k$ Hadamards, incoherent resources and an arbitrary ancilla to implement $n>k$ Hadamards \textit{approximately} or \textit{probabilistically}, which involves considering subchannels as in \cref{eq:general_subchannel}. We leave these questions to ongoing and future work, in order to complete the picture as presented in \cref{tab:results}.

In the case of using incoherent resources to simulate $n$ Hadamards, we were content to show that approximate implementation is not possible, and we have left the optimality of our bound open. Specifically, we were able to exploit the fact that the corresponding channels commuted with the dephasing map. For channels that use a non-zero amount of coherence (e.g. using $k<n$ Hadamards), one possibility would be to find a similar characterisation, for example, commutation with some channel that only allows a small amount of coherence through. 

It would also be interesting to understand if there is any advantage at all to using an ancillary state. We proved a much weaker lower bound (\cref{lem:0ton_approx}) compared to not using an ancilla at all (\cref{lem:ancilla_indep}), perhaps there is still some advantage to be had here. In this work, we were primarily concerned with providing lower bounds in order to show no-go results, but do there exist interesting upper bounds? That is, perhaps one could show that using an ancillary state allows for a strictly better approximation to a coherent unitary, compared to the case of no ancilla.

\begin{problem}
Improve or show optimality of the bounds presented in this work, and find lower bounds on implementing $n$ Hadamards using $k$ Hadamards, incoherent unitaries, classical control, computational basis measurements, and an arbitrary ancilla.
\end{problem}

\subsection{Links with MBQC } \label{subsec:discussmbqc}

Our results concern the circuit model, so it would also be interesting to extend our results to the measurement based scenario. Let us first review some previous works, before commenting on how one could formulate and engage with analogous questions here.

The universality of states for one-way MBQC has been studied in \cite{van2007fundamentals, mora2010universal}, taking LOCC as the free operations. In particular, in \cite{van2007fundamentals} they distinguish between four types of universality depending upon whether the input and output considered are classical ({\bf{C}}) or quantum ({\bf{Q}}). In this language, a device is considered to be {\bf{QQ}}-universal if it can implement any unitary operation $U$, and {\bf{CQ}}-universal if it can prepare any pure quantum state $\ket{\psi}$. This scenarios are natural when respectively considering the circuit model and the measurement based model (for which the local operations exclude the notion of a quantum input). The authors also discuss the subtleties of efficiency, and of approximate and probabilistic universality in \cite{mora2010universal}.

As {\bf{QQ}}-universality (quantum inputs and outputs) is concerned with simulation of a unitary channel, it possess a parallel with the gadget-based approach considered in this work. On the other hand, as {\bf{CQ}}-universality is the appropriate notion for MBQC, there are some intricacies involved, for example one must take into account the dimension of the ancillary state for any meaningful definition. To see this, recall that a $\epsilon$-net is a set of pure quantum states such that \textit{any} pure quantum state is within distance $\epsilon$ of some state in the net. Hence, one could encode such an $\epsilon$-net into an ancilla (i.e. take the tensor product of all states in the net). Then by simply tracing out all but one of the subsystems, one could prepare any pure state to within $\epsilon$ distance. However the size of any $\epsilon$-net increases rapidly with the dimension of the systems \cite{hayden2004randomizing, montanaro2019quantum}, hence this approach is highly impractical and more refined ideas would be needed.

One could extend the definitions in \cite{van2007fundamentals, mora2010universal} to more general free operations by posing the following question: given some set of allowed operations (e.g. LOCC) does there exists a family of resourceful ancillary states that yield $\mathbf{CQ}$-universality? We provide an example of such a definition here, inspired by the aformentioned works.

\begin{definition}
A family of sets of operations $\{\mathcal{V}_n\}$ is \textit{$\epsilon$-approximate, efficiently}  $\mathbf{CQ}$-universal, with respect to a distance measure $D$ if:
\begin{align*} 
\text{there exists:\qquad}  &\text{a family of states $\{ \ket{\ancillapure(n)} \}$, where each $\ket{\ancillapure(n)}$ is on at most poly$(n)$ qubits} \\
\text{such that:\qquad} &\text{for every family of states $\{ \ket{\psi_n} \in (\mathbbm{C}^{ 2})^{\otimes n}$ \}  obtainable by a uniform } \\ &\quad \text{ family of quantum circuits of depth at most poly$(n)$ } \\
\text{there exist:\qquad} &\text{maps $\mathcal{E}_n \in \mathcal{V}_n$}\\
\text{such that:\qquad} &\text{$D \bigg (\mathcal{E}_n( \ketbra{\ancillapure(n)}), \ketbra{\psi_n} \bigg ) \leq \epsilon \quad \forall n$}.
\end{align*}
\end{definition}

One could also consider the probabilistic case, see \cite{mora2010universal}. For example, the operations LOCC are universal under this definition, as the cluster states serve as the resource state. Similarly, due to the results in \cite{takeuchi2019quantum}, the ability to only perform local Hadamard gates (with adaptivity and computational basis measurements) are universal using hypergraph states.

Hence a natural extension of our work to the MBQC framework could be to consider if for \textit{incoherent} LOCC operations there exists a family of resourceful states such that the pair is approximately, efficiently universal. More specifically, one could ask if there exists a family of resource states such that one can achieve efficient, universal quantum computation using only computational basis measurements (at least two measurement bases may at first seem to be necessary \cite{takeuchi2019quantum}, however note that some adaptivity is still possible in the order of systems measured). Furthermore, what would the analogous version of the $k \mapsto n$ question be? Given the ability to perform only $k$ Hadamards (or $k$ X measurements), can one perform universal quantum computation? Again this may translate to the existence of some fixed resourceful state, from which any state could be prepared.

\begin{problem}
    Is efficient, universal measurement based quantum computation possible with only incoherent resources (for example, using only $Z$ measurements)?
\end{problem}

The above points also raise interesting questions about the relationship between coherence and quantum incompatibility \cite{guhne2021incompatible, heinosaari2016invitation}. The latter  refers to the fact that not all measurements can be simultaneously performed in quantum theory. For projective measurements, the natural condition is whether the corresponding observables commute (as then a common eigenbasis exists to measure in). For more general POVM measurements, the prevailing notion is to ask for the existence of a so-called \textit{parent} measurement, from which the outcomes of all other measurements can be post-processed\footnote{Equivalently, one can consider the commutativity of the Naimark dilation of the measurements.} \cite{guhne2021incompatible}.

Within our framework, we related the ability to perform in different measurement bases to the unitary that maps between the bases (i.e. in the Heisenberg picture). This could indicate that these resources are equivalent in some way, and perhaps that for any model of universal quantum computation \textit{either} coherence must be present in the operations, or some form of incompatibility must be present in the measurements.
 
\begin{problem}
    Are coherence and measurement incompatibility computationally related?
\end{problem}

We also remark that due to existing Hadamard gadgets \cite{heyfron2018efficient, jozsa2013classical, Beaudrap2020FastAE},  the ability to measure in the $X$ basis is equivalent to the ability to implement the Hadamard, given the ability to perform incoherent unitaries freely and access to ancillas.

\subsection{Resource Theories}

Finally, our analysis raises some interesting questions for general resource theories \cite{chitambar2019quantum}, see \cref{app:resourcecoherence} for some background.

For example, \cite{diaz2018using} studies using maximally incoherent operations (MIOs) to implement arbitrary channels, in terms of the resourcefulness (e.g. coherence rank) of supplementary ancillas. MIOs are exactly the channels that map the set of incoherent states to themselves. It would be interesting to study if this could give rise to a novel model of quantum computation, using MIOs acting on resourceful coherent states.  It was also shown in \cite{lami2019generic} that the coherence distillation capabilities of strictly incoherent operations (SIO) and physically incoherent operations (PIO) are very limited.

\begin{problem}
    Where else can the ``cut'' be placed? Are there interesting new models of computation?
\end{problem}

 Furthermore, one way of seeing how our result went through is that for the resource theory of coherence, the quantum controlled free unitaries are also free. This is not the case for e.g. magic, as an $S$ gate is Clifford, but a controlled $S$ gate is not Clifford, or for LOCC (e.g. CNOT). This leads to a natural question: 
\begin{problem}
    Are there other resource theories for which taking quantum control of the free operations remains free?
\end{problem}

Our results also hint at a general trade-off between resource generating power and unitarity. Consider a free set of states $\mathcal{F}$ and a resource quantifier $Q$, and define the resource generating power of a channel $\mathcal{V}$ as $\max_{\rho \in \mathcal{F}} Q( \mathcal{V}(\rho))$. Suppose a channel is of the form
\[
\mathcal{E}(\rho) = \text{Tr}_2 \left (U \rho \otimes \ancillamixed U^\dagger \right )
\]
and $U$ has resource generating power $\alpha$.  If the channel $\mathcal{E}$ has resource generating power strictly greater than $\alpha$, intuitively this might suggest that $U$ is swapping in some of the resource contained in $\ancillamixed$, and hence $\mathcal{E}$ must be compromising on being unitary. Clearly the total resource content of $\mathcal{E}$ should be somehow upper bounded by the sum of that of $U$ and the state $\ancillamixed$. But in order for $\mathcal{E}$ to be unitary, perhaps it cannot use any of the resource contained in $\ancillamixed$. 
See \cite{takagi2020universal} for related work in this direction. The authors provide quantitative relations between resource content, implementation accuracy, and the dimension of the ancillary system, which they show diverges as the implementation accuracy goes to zero.

A similar setup is also considered in \cite{chiribella2021fundamental}. In particular, the authors consider a channel involving free unitaries acting on a resourceful state in order to implement a resourceful channel. They prove a lower bound on the resource content of the ancillary state as a function of the resource content of the target channel, and also apply this result in the context of coherence.

\begin{problem}
    Are there trade-offs between unitarity and resource generating power?
\end{problem}

In light of this, we note that the style of channels considered in this work seems to hint at a potential new class of resourceful operations. We essentially consider resourceless channels with access to an arbitrarily resourceful state, which does not neatly fit into existing resource theoretic frameworks (see \cref{rem:pio_channels} in \cref{app:resourcecoherence}).

\subsection{Concluding Remarks}

Whist progress has been made in understanding the components required to achieve a super-classical speedup, such as entanglement and magic, there are many exciting research avenues open to explore. The subfield of quantum resource theories has had relatively little intersection with topics in quantum computation, and there may be much to be gain from approaches that attempt to unify the different models of computation, such as gadget or measurement based. It is our hope that through studying characteristic features of quantum theory (such as coherence) on the level of states, channels and measurements, one may hope to gain a more complete understanding of the power of quantum computation.

\subsection*{Data Access Statement}
There was no data generated as part of this research.

\section*{Acknowledgements}
We are extremely grateful to Andreas Winter for pointing out that the previous \cref{lem:0ton_approx} could be significantly improved, and for kindly letting us use the new version of the Lemma and his proof.

We also thank members of the Bristol Quantum Information Theory group for helpful discussions. BDMJ acknowledges support from UK EPSRC (EP/S023607/1) . PS is a CIFAR Azrieli Global Scholar in the Quantum Information Science Program, and also acknowledges support from a Royal Society University Research Fellowship (UHQT/NFQI). NL gratefully acknowledges support from the UK Engineering and Physical Sciences Research Council through Grants No. EP/R043957/1, No. EP/S005021/1, and No. EP/T001062/1.

\addcontentsline{toc}{section}{References}
\bibliographystyle{unsrtnat}
\bibliography{references}

\appendix

\section{Resource Theories and Coherence} \label{app:resourcecoherence}

Quantum resource theories \cite{chitambar2019quantum} are flourishing as an active area of research. The primary goal is to consider unifying principles across different aspects of quantum mechanics that are quintessentially `quantum'. Specific examples include entanglement \cite{horodecki2009quantum}, coherence \cite{streltsov2017colloquium}, magic \cite{howard2017application, seddon2019quantifying}, and incompatibility \cite{guhne2021incompatible}. There are multiple axiomatic approaches: one can either start with some well-motivated free set of states and define the free channels as those preserving this set, or start with operationally motivated free operations (e.g. LOCC) and define the free states as those which can be generated using free operations alone.

For coherence, the starting point is to fix some particular basis $\{ \ket{x} \}$ as ``free'', and refer to this basis as \textit{incoherent}. These basis states can be thought of as easy to prepare, and in our work we consider them as computational basis states. An incoherent pure state is then equal to a single one of these basis states, and superpositions or coherent states are considered resourceful.

Formally, an arbitrary mixed state $\rho$ is called \textit{incoherent} with respect to the basis $\{ \ket{x} \}$ if it can be written as
\[
\rho = \sum_x p_x \ketbra{x},
\]
for some probabilities $p_x$, i.e. it is diagonal in this basis. Conceptually this consists of all the states that can be written as probabilistic mixtures of computational basis states, with no superposition present. Note that the maximally mixed state $\frac{\id}{d}$ is an example of such an incoherent state (with $p_x = \tfrac{1}{d} \quad \forall ~ x)$. We refer to the set of incoherent states as $\mathcal{I}$.

A unitary $U$ is \textit{incoherent} relative to the basis $\{ \ket{x} \}_{x=1}^d$ if it can be written as
\[
U = \sum_{x=1}^d e^{i \theta_x} \ketbra{\pi(x)}{x} \label{eq:incoherent_unitary}
\]
for some string of $d$ real numbers $\theta_x$ and some permutation $\pi$ on $d$ elements. In particular, incoherent unitaries map a computational basis state to another computational basis state, possibly multiplied by some phase. This definition also implies that $U\rho U^\dagger \in \mathcal{I}$ if $\rho \in \mathcal{I}$.

\begin{example}
Examples of incoherent unitaries include the Pauli matrices, the phase and $T$ gates, CNOT, SWAP, and the Toffoli gate. Examples of unitaries that are coherent (i.e. not incoherent, able to generate coherence) include the Hadamard gate, the Fourier transform, and $X$ rotations $e^{i\theta X}$ for $\theta \notin \{n \pi ~ ; ~ n \in \mathbbm{Z} \}$.
\end{example}

\begin{remark}
Note that if a unitary cannot create any superpositions, then it must be of the form
\[
U = \sum_x \alpha_x \ketbra{\pi(x)}{x}
\]
for some complex numbers $\alpha_x$. However for this to be unitary, we must have that $\abs{\alpha_x} = 1$ for all $x$. Hence if a unitary is not of the form \cref{eq:incoherent_unitary}, then it must necessarily map at least one computational basis state to a superposition (linear combination) of at least two computational basis states (i.e. it cannot change the magnitude of a computational basis state).

\end{remark}

For a fixed basis $\ket{x}$, the dephasing map is defined as 
\[
\Delta (\rho) := \sum_x \ketbra{x} \rho \ketbra{x}
\]
this has the effect of removing the off-diagonal elements on a density matrix, and is a valid quantum channel (it is trace-preserving and completely positive).

There are many different approaches to defining a free set of operations in this resource theory, see \cite{chitambar2016comparison, streltsov2017colloquium} for summaries. We review several of them here. \textit{Maximally Incoherent Operations} (MIO) map incoherent states to other incoherent states, namely $\mathcal{E}$ is a MIO if $\mathcal{E} (\mathcal{I}) \subseteq\mathcal{I} $.

\begin{lemma}
A channel $\Omega$ maps incoherent states to incoherent states (i.e. is MIO) if and only if $\Delta \circ \Omega \circ \Delta = \Omega \circ \Delta$.
\end{lemma}

\begin{proof}
Note that for all $\rho \in \mathcal{I}$ we have $\Delta(\rho) = \rho$. If $\Omega$ maps incoherent states to incoherent states, then we must have
$\Omega (\Delta (\rho)) = \Delta ( \Omega (\Delta (\rho)) )$ for all $\rho$, which implies $\Delta \circ \Omega \circ \Delta = \Omega \circ \Delta$.

To show the other direction, assume that $\Delta \circ \Omega \circ \Delta = \Omega \circ \Delta$. Then for $\rho$ incoherent, we have $\Omega (\rho) = \Omega (\Delta (\rho)) =\Delta ( \Omega (\Delta (\rho)) ) $, which is incoherent.
\end{proof}

The set of \textit{incoherent operations} (IO) is defined as the set of quantum channels $\mathcal{E}$ which admit a Kraus decomposition $\mathcal{E} (\rho) = \sum_\lambda K_\lambda ~ \rho ~ K_\lambda^\dagger$ such that $\frac{K_\lambda \rho K_\lambda^\dagger}{\text{Tr}(K_\lambda \rho K_\lambda^\dagger)} \in \mathcal{I}$ for all $\rho \in \mathcal{I}$. This definition means that it is not possible to generate coherence even probabilistically given access to the quantum instrument defined by $\{ K_\lambda \}$. This definition is equivalent to being able to write each $K_\lambda$ in the form $ \sum_x \alpha_x \ketbra{\pi(x)}{x}$ where the coefficients $\alpha_x$ can be arbitrary complex numbers. If each $K_\lambda^\dagger$ can also be written this way, the corresponding operations are referred to as \textit{strictly incoherent operations} (SIO). We also mention \textit{physically incoherent operations} (PIO), which are channels that can be realised via performing a global incoherent unitary on the input state and some incoherent ancilla, followed by an incoherent measurement and classical processing \cite{chitamber2016critical}. Finally, \textit{dephasing-covariant incoherent operations} (DIO) are channels $\mathcal{E}$ which commute with the dephasing map: $\mathcal{E} \circ \Delta = \Delta \circ \mathcal{E}$. We have the following inclusions \cite{chitamber2016critical}
\[
PIO \subseteq SIO \subseteq DIO \subseteq MIO.
\]
\begin{remark} \label{rem:pio_channels} 
   We are considering channels of the following form, for $U$ incoherent,
    \[
    \mathcal{E}^x(\rho) = \text{Tr}_X \bigg ( \id \otimes \ketbra{x} U \rho U^\dagger \bigg ).
    \]
     These fall within the class $PIO$, but note that channels such as $\rho \mapsto \mathcal{E}^x(\rho \otimes \ancillamixed)$ do not, as we can SWAP in the ancillary system (which in our framework could be arbitrarily resourceful).
\end{remark}

\end{document}